\documentclass[a4paper,12pt]{amsart}
\usepackage[margin=2cm]{geometry}
\usepackage{color}
\usepackage{pdfsync}
\usepackage{amsmath,cool}

\newtheorem{thm}{Theorem}[section]
\newtheorem{cor}[thm]{Corollary}

\newtheorem{prop}[thm]{Proposition}
\theoremstyle{definition}
\newtheorem{defn}[thm]{Definition}
\theoremstyle{remark}
\newtheorem{rem}[thm]{Remark}
\numberwithin{equation}{section}
\newcommand{\R}{\mathbb R}

\newcommand{\norm}[1]{\left\Vert#1\right\Vert}
\newcommand{\abs}[1]{\left\vert#1\right\vert}
\newcommand{\set}[1]{\left\{#1\right\}}

\newcommand{\To}{\longrightarrow}

\newcommand{\cald}{\mathcal D}

\newcommand{\pr}{^{\prime}}

\newcommand{\beq}{\begin{equation}}
\newcommand{\eeq}{\end{equation}}
\newcommand{\ben}{\begin{enumerate}}
\newcommand{\een}{\end{enumerate}}

\newcommand{\C}{\mathbb C}

\newcommand{\lan}{\langle}
\newcommand{\ran}{\rangle}
\newcommand{\bs}{\boldsymbol}

\begin{document}

\title[]{Parametrizations of degenerate density matrices}

\author{E. Br\"uning \and S. Nagamachi}

\address[E. Br\"uning]{School Mathematics, Statistics and Computer Science,
University of KwaZulu-Natal, Private Bag X54001, Durban 4000,
South Africa, and Natonal Institute for Theoretical Physics (NITheP), KwaZulu-Natal, South Africa} \email{bruninge@ukzn.ac.za}

\address[S. Nagamachi]{Emeritus Professor, The University of Tokushima\\ Tokushima 770-8506,
	Japan} \email{shigeaki-nagamachi@memoad.jp}

\begin{abstract}
It turns out that a parametrization of degenerate density matrices requires a parametrization of 	
$\mathfrak{F}=U(n)/({U(k_1)\times U(k_2)\times \cdots
	\times	U(k_m)})\quad n=k_1 +\cdots + k_m $ where $U(k)$ denotes the set of all unitary $k\times k$-matrices with complex entries. Unfortunately the parametrization of this quotient space is quite involved. Our solution does not rely on Lie algebra methods {directly,} but succeeds through the construction of suitable sections for natural projections, by using techniques from the theory of homogeneous spaces. 
We mention the relation to the Lie algebra back ground and conclude with two concrete examples.
\end{abstract}

\maketitle \tableofcontents

\section{Introduction} \label{intro}
In various parts of Physics density matrices, i.e., positive trace class operators  of trace $1$ on a complex separable Hilbert space play an important role, see \cite{Bl12}. Density matrices represent states of quantum systems. 
In many concrete applications the Hilbert space is typically finite dimensional and the Hilbert space then is $\C^n$, the space of $n$-tuples of complex numbers with its standard inner product. Thus the space $\cald_n$ of all density matrices on  $\C^n$ is the space of all $n\times n$ matrices $\rho$ with complex entries such that $\lan x,\rho x \ran \geq 0$ for all
$x \in \C^n$ and $\Tr(\rho)= \sum_{j=1}^{n} \lan e_j,\rho e_j\ran = 1$
for any orthonormal basis $\set{e_j:\, j = 1,\ldots,n}$ of $\C^n$.
Through these two constraints the entries of a density matrix are not all independent and thus contain redundant parts. But for an effective description of quantum states one would like to get rid of these redundant parts of a density matrix, i.e., one would like to have a description of density matrices in terms of a set of independent parameters, that is a parametrization in the sense of
Definition \ref{def:parametrization}. The best known parametrization of density matrices seems to be the Bloch vector
parametrization \cite{Bl46,NC04}. While this parametrization is perfect for $n=2$-level systems, it has a serious defect for $n\geq 3$-level systems in the sense that the parameter set cannot be determined explicitly (see for instance \cite{BMMP12}). Thus various authors have been looking for alternative ways to parametrize density matrices, see for instance \cite{Ak07,BK03,Harriman78a,Spengler10}. 
Some time ago we started with a parametrization of density matrices based on their spectral representation \cite{BP08,BCP08,BMMP12}.

The spectral representation of a density matrix  $\rho \in \cald_n$ reads
\beq \label{eq:spect-rep}
\rho =U D_n(\lambda_1, \ldots,\lambda_n)U^*
\eeq
where $D_n(\lambda_1, \ldots,\lambda_n)$ is the diagonal matrix of the eigenvalues $\lambda_1, \ldots,\lambda_n$ and $U$ is some
unitary $n\times n$ matrix, i.e., $U\in U(n)$.
These $n$ eigenvalues are not necessarily distinct; they occur in this list as many times as their multiplicity requires.

In this article we consider parametrizations in the strict sense as suggested in \cite{BMMP12}. This definition reads:
\begin{defn} \label{def:parametrization}
	A {\bf { parametrization of  density matrices}}  is given by the following:
	\begin{enumerate}
		\item[(a)] Specification of a parameter set $Q_n \subset \R^m$ where $m$ depends on $n$, i.e., $m=m(n)$;
		\item[(b)] Specification of a one-to-one and onto map $F_n:\, Q_n \To \cald_n$.\\
	\end{enumerate}
\end{defn}

When the spectral representation (\ref{eq:spect-rep}) is chosen as the starting point one obviously needs a suitable parametrization of unitary matrices.

The set of eigenvalues $\set{\lambda_1,\ldots,\lambda_n}$ can be ordered according to their size: We denote the set of eigenvalues ordered in this way by $\Lambda_n$, i.e.,
\beq \label{eq:ev}
\Lambda_n = \set{\bs{\lambda}=(\lambda_1,\ldots,\lambda_n):
	0\leq \lambda_n \leq \cdots \leq \lambda_2 \leq \lambda_1, \sum_{j=1}^n \lambda_j =1}	\;.
\eeq

We begin by addressing the question of uniqueness of the spectral representation (\ref{eq:spect-rep}).  Accordingly suppose that for $\bs{\lambda},\bs{\lambda}\pr \in \Lambda_n$ and
$U,V \in U(n)$ we have
$$ U^* D_n(\bs{\lambda})U =V^*D_n(\bs{\lambda}\pr)V $$
Since the spectrum of a matrix is uniquely determined and since $\bs{\lambda},\bs{\lambda}\pr \in \Lambda_n$ it follows $\bs{\lambda}=\bs{\lambda}\pr$ and therefore it follows $V U^* D_n(\bs{\lambda})=D_n(\bs{\lambda})VU^*$, i.e.,
\beq \label{eq:commut}
V U^* \in D_n(\bs{\lambda})\pr
\eeq
where $D_n(\bs{\lambda})\pr$ denotes the commutant of the diagonal matrix $ D_n(\bs{\lambda})$ in $U(n)$. If a density matrix $\rho \in \cald_n$ has a non-degenerate spectrum, i.e., if 
\beq \label{eq:nondeg-spec}
\bs{\lambda} \in \Lambda_n^{\neq} =\set{\bs{\lambda}\in \Lambda_n:
	0\leq \lambda_n <\lambda_{n-1}< \cdots < \lambda_2 <\lambda_1}
\eeq
then this commutant is easily determined and is given by
\beq \label{eq:comm1}
\cald_n(\bs{\lambda})\pr =U(1)\times \cdots \times U(1),\quad n\;\rm{terms}
\eeq
Naturally there are many ways in which a density matrix can be degenerate. Suppose that the spectrum of $\rho_n \in \cald_n$ has $m$ different eigen-values $\lambda_1, \ldots , \lambda_m$ with multiplicities $k_1,\ldots, k_m$ with 
$\sum_{j=1}^m k_j =n$. Thus $\bs{\lambda} \in \Lambda_n$
is of the form
$$(\lambda_1, \ldots, \lambda_1, \lambda_2,\ldots, \lambda_2, \ldots \lambda_m, \ldots \lambda_m) $$
where each $\lambda_j$ is repeated $k_j$ times and 
$\sum_{j=1}^m k_j \lambda_j =1$ and where we use the ordering $0\leq \lambda_m <\lambda_{m-1} < \cdots < \lambda_1$ according to (\ref{eq:ev}). Thus one has in this case
\beq \label{eq:diag-degen}
D_n(\bs{\lambda}) = \rm{diag}_n(
\lambda_1 I_{k_1},\lambda_2 I_{k_2},\ldots,
\lambda_m I_{k_m} ) \eeq
where $\rm{diag}_n$ denotes the $n \times n$ diagonal matrix with entries as indicated and where $I_{k_j}$ denotes the $k_j \times k_j$ identity matrix. Therefore the commutant of the diagonal matrix $D_n(\bs{\lambda})$ is in this case
\beq \label{eq:commut-degenerate}
D_n(\bs{\lambda})\pr =U(k_1)\times U(k_2)\times \cdots \times U(k_m)\;.
\eeq

Thus in order to complete the parametrization problem for
degenerate density matrices we need to find a suitable parametrization of 
\beq \label{eq:deg-param} 
\mathfrak{F}=U(n)/({U(k_1)\times U(k_2)\times \cdots
	\times	U(k_m)})\quad n=k_1 +\cdots + k_m .
\eeq

We begin with a discussion of the simplest case, i.e., $k_j=1$ for all $j$ and $m=n$.
Note that in this case
$$U(n)/(U(1)\times \cdots \times U(1))=U(n)/\sim $$
for the equivalence relation $\sim$ in $U(n)$ defined by
$$      U \sim  U^{\prime } \Leftrightarrow  U^{-1}U^{\prime } \in  U(1) \times  \ldots \times  U(1).$$
Accordingly the elements of $U(n)/\sim$ are the equivalence classes
$$      [U] = \{ UV; V \in  U(1) \times  \cdots \times  U(1)\} $$

$U(n)/\sim =U(n)/(U(1) \times  \cdots \times  U(1))$ is called a (complex) full flag
manifold (see \cite{Ma72}).  Introduce the natural projection
$$      \pi : U(n) \ni  U \rightarrow  [U] \in  U(n)/(U(1) \times  \cdots \times  U(1)).$$
Then a map $\iota : U(n)/(U(1) \times  \cdots \times  U(1)) \rightarrow  U(n)$ is called a
{\bf section of $U(n)$ on $U(n)/(U(1) \times  \cdots \times  U(1))$} for
$\pi$ if $\pi \circ \iota  = {\rm id \, }$.
The relation $U^{\prime } = \iota ([U])$ implies that $U^{\prime }$ is a representative of
the coset
$[U]$.  The mapping
$$      p: \Lambda_{n}^{\neq} \times  U(n)/(U(1) \times  \cdots \times  U(1)) \ni  ((\lambda _{1}, \ldots , \lambda _{n}), m)
\longrightarrow  \iota (m) D_{n}(\lambda _{1}, \ldots , \lambda _{n}) \iota (m)^{*}$$
does not depend on the section $\iota $, and thus the mapping $p$ gives
a parametrization of density matrices, if $U(n)/(U(1) \times  \cdots \times 
U(1))$ is suitably parametrized.  Since $U(n)/(U(1) \times  \cdots \times 
U(1))$ is a manifold, it is parametrized locally.  But unfortunately,
this parametrization is not simple.  Even though the mapping does
not depend on $\iota $, the construction of a concrete section is
necessary, but also not so simple.

Through the construction of a concrete section we will also achieve a parametrization of unitary matrices, an important problem in itself which has found
considerable attention in the last $10 - 15$ years (see the references mentioned above).
The starting point of this construction is
the so called canonical coset decomposition which gives in particular the well-known Jarlskog parametrization \cite{Ja05,Ja06}.

Recall that the  coset space $U(n)/(U(n-1) \times  U(1))$ is the projective space $CP^{n-1}$ (see \cite{Ma72}).

Symbolically, the canonical coset decomposition is:
$$      U(n) = U(n)/(U(n-1) \times  U(1)) \cdot  U(n-1)/(U(n-2) \times  U(1))$$
$$      \cdot  \ldots \cdot  U(2)/(U(1) \times  U(1)) \cdot  (U(1) \times  \ldots \times  U(1)).$$
In Section \ref{projectsp}, we parametrize $U(n)$ by constructing sections $\iota _j: CP^{n-j} \rightarrow U(n-j+1)$ for $j = 1, \ldots , n-1$.

For the degenerate case we have only to use (complex) Grassmann
manifolds $U(k)/(U(k_{1}) \times  U(k_{2})$ instead of the projective spaces
$U(k)/(U(k-1) \times  U(1))$:
$$      U(n) = U(n)/(U(n-k_{m}) \times  U(k_{m})) \cdot  U(n-k_{m})/(U(n-k_{m-1}-k_{m}) \times  U(k_{m-1}))$$
$$      \cdot  \ldots \cdot  U(k_{1}+k_{2})/(U(k_{1}) \times  U(k_{2})) \cdot  (U(k_{1}) \times  \cdots \times  U(k_{r})).$$
In Section \ref{grassmann}, we study  this case extensively because the parametrization of degenerate density matrices is the main new result of this paper (Propositons \ref{gr:canonicaldec}, \ref{gr:flagsection}).  The result of Section \ref{projectsp} is the special case of $k_j = 1$ for $1 \leq j \leq n$.
In Section \ref{jarlskog}, we study Jarlskog parametrization used in \cite{BMMP12} for the non-degenerate case.

If $S(\C^n) $ denotes the unit sphere in $\C^n$,  we can parametrize the subset $\Omega  = \{ [z]; z \in  S(\mathbb C^{n}), \  z_{n} \neq
0\}  \subset  CP^{n-1}$ by $B(n-1) = \{ x \in  \mathbb C^{n-1}; \Vert x\Vert  < 1\} $.  But for the boundary
$\partial B(n-1) = \{ x \in  \mathbb C^{n-1}; \Vert x\Vert  = 1\} $, the mapping 
$$      \partial B(n-1) \ni  x \rightarrow  [x] \in  CP^{n-1}$$
is not injective and consequently there are $x$ and $x^{\prime }$ in $\partial B(n-1)$
such that $W(x) \neq W(x^{\prime })$ for $W(x)$ of (\ref{pr:Wofx}) and there exist $V$ and $V^{\prime }$ in $U(n-1)
\times  U(1)$ such that
$$       W(x)V = W(x^{\prime })V^{\prime }.$$
Consequently, the parametrization of density matrices is not always unique.
{
	In Section \ref{section} we present a way to construct a section on a Grassmann manifold by using sections on a suitable projective space, since, for concrete calculations, the construction of a section presented in Section \ref{grassmann} is fairly involved.  In Section \ref{example}, we give simple concrete examples of degenerate density matrices.
}
In this paper, we mainly use the technique of homogeneous spaces.  But there is the theory of Lie algebra behind it.  In Section \ref{liealg}, a Lie algebraic back ground is presented.
\section{Grassmannian and canonical coset decomposition} \label{grassmann}

The Grassmann manifold $G(k, \mathbb C^{n})$ is the set of all complex
$k$-dimensional subspaces of $\mathbb C^{n}$ (see \cite{Ma72}).  Let $W$ be a $k$-dimensional
subspace of $\mathbb C^{n}$.  Then we choose a basis of column vectors $\boldsymbol{z}_{1}, \ldots \boldsymbol{z}_{k}$ of $W$
and associate with it the matrix
$$      M(W) = \begin{pmatrix} \boldsymbol{z}_{1} & \boldsymbol{z}_{2} & \ldots & \boldsymbol{z}_{k} \end{pmatrix}.$$
Since the matrix $M(W)$ depends on the choice of a basis of $W$,
$M(W)$ is not determined uniquely by $W$.  There is a freedom of
multiplication by regular $k\times k$  matrices from the right.  Thus we have
$$      G(k, \mathbb C^{n}) = \{ M;M= {\rm complex } \; n\times k \  {\rm matrix \, }\,  {\rm of \, }\,  {\rm rank \, }\,  k\} /GL(k, \mathbb C).$$
Let $S_{k, n}$ be the set of permutations $\sigma $ of $\set{1,\ldots,n}$ such that $1 \leq  \sigma (1) < \sigma (2) <
\cdots < \sigma (n-k) \leq  n$ and $1 \leq  \sigma (n-k+1) < \sigma (n-k+2) < \cdots < \sigma (n) \leq 
n$.
Let $M(W)_{\sigma (n-k+1), \sigma (n-k+2), \cdots , \sigma (n)}$ be the $k\times k$-matrix
which consists of the $\sigma (n-k+1), \sigma (n-k+2), \ldots , \sigma (n)$-th rows of $M(W)$,
and define
$$    \Omega _{\sigma } = \{ W \in  G(k, \mathbb C^{n}); \det M(W)_{\sigma (n-k+1), \sigma (n-k+2), \ldots , \sigma (n)} \neq 0\} \subset G(k, \mathbb C^{n}).$$
Since the rank of the matrix $M(W)$ is $k$, there is a $\sigma  \in  S_{k, n}$
such that \\
$\det M(W)_{\sigma (n-k+1), \sigma (n-k+2), \ldots , \sigma (n)} \neq 0$.  Thus
we have
\beq  \label{gr:union}
G(k, \mathbb C^{n}) = \cup _{\sigma \in S} \Omega _{\sigma }.
\eeq
Let $M(n-k, k)$ be a set of all $(n-k)\times k$ complex matrices.
Define the mapping
\beq \label{gr:chart}
\phi _{\sigma }: \Omega _{\sigma } \ni  W \rightarrow  M(W)_{\sigma (1), \ldots , \sigma (n-k)}
M(W)_{\sigma (n-k+1), \ldots , \sigma (n)}^{-1} \in  M(n-k, k). 
\eeq     
Then $\phi _{\sigma }$ gives the homeomorphism
$$      \Omega _{\sigma } \cong  M(n-k, k) \cong  \mathbb C^{(n-k)k}.$$
In fact,
the element $W$ of $\Omega _{\sigma }$ corresponds in a 1-1 way to the
matrix of the form
\beq \label{gr:Z}
M(W)_{\sigma } M(W)_{\sigma {n-k+1}, \ldots , \sigma {n}}^{-1}
= \begin{pmatrix} \phi _{\sigma }(W) \cr I \end{pmatrix}      
= \begin{pmatrix} Z \cr I \end{pmatrix}      
= \begin{pmatrix} z_{11} & \ldots & z_{1k} \cr
	\vdots & \ddots & \vdots \cr
	z_{n-k,1} & \ldots & z_{n-k,k} \cr
	1   & \ldots & 0 \cr
	\vdots & \ddots & \vdots \cr
	0   & \ldots & 1  \end{pmatrix} ,
\eeq 
where $M(W)_{\sigma }$ is the matrix whose $i$-th row is the $\sigma (i)$-th row of $M(W)$, $Z \in  M(n-k, k)$ and $I$ is the $k\times k$ identity matrix.  Then
the set $\{ (\Omega _{\sigma }, \phi _{\sigma }); \sigma  \in  S_{k,n}\} $ gives an atlas of $G(k, \mathbb C^{n})$.

There is another parametrization of $\Omega _{\sigma }$ which is more
convenient for our purpose.\newline
Let
$$      B(n-k, k) = \{ X \in  M(n-k, k); X^{*}X < I_k\} .$$
Then the set $\Omega _{\sigma }$ can be parametrized by the set $B(n-k,k)$.  We will show this in the following.

Let $S(k, \mathbb C^{n})$ denote the set of all  orthonormal frames $F =( \boldsymbol{x}_{1},  \boldsymbol{x}_{2}, \cdots , \boldsymbol{x}_{k}) $  in $\mathbb C^{n}$ of length $k$, i.e.,$\lan \bs{x}_j,\bs{x}_i \ran = \delta_{ij}$ for $i,j=1,\ldots,k$.  For
$F \in  S(k, \mathbb C^{n})$ we associate a matrix $M(F)$ by\newline
$$    M(F) = ( \boldsymbol{x}_{1},\boldsymbol{x}_{2}, \ldots , \boldsymbol{x}_{k} )
.$$ 
Let $M(F)_{\sigma (n-k+1), \sigma (n-k+2), \ldots , \sigma (n)}$ ($F \in  S(k, \mathbb C^{n})$)
be the $k\times k$-matrix which consists of $\sigma (n-k+1), \sigma (n-k+2), \ldots ,
\sigma (n)$-th rows of $M(F)$,
and let
\beq \label{eq:omegatilde}
\tilde{\Omega } _{\sigma } = \{ F \in  S(k, \mathbb C^{n}); \det
M(F)_{\sigma (n-k+1), \sigma (n-k+2), \ldots , \sigma (n)} \neq 0\}  \subset  S(k, \mathbb C^{n}). \eeq
Since the rank of the matrix $M(F)$ is $k$, there is a
$\sigma  \in  S_{k, n}$ such that \newline
$\det M(F)_{\sigma (n-k+1), \sigma (n-k+2), \ldots , \sigma (n)} \neq 0$.  Thus
we have
$$      S(k, \mathbb C^{n}) = \cup _{\sigma \in S_{k, n}} \tilde{\Omega } _{\sigma }.$$
\begin{defn} \label{gr:usigma} For a permutation $\sigma$ of $\set{1,2,\ldots,n}$ define $U_{\sigma } \in  U(n)$ by $U_{\sigma }e_{j} = e_{\sigma (j)}$.  
\end{defn}
Then $M(F)_{\sigma } = U_{\sigma }^{-1}M(F)$.  Define $F_{\sigma }$ by $M(F_{\sigma }) = M(F)_{\sigma }$, and identify $F$ and $M(F)$.
Let $\pi _{2}$ be the surjective mapping
\beq  \label{gr:pi2}
\pi _{2}: S(k, \mathbb C^{n}) \ni  F \rightarrow  W = {\rm span \, }F \in  G(k, \mathbb C^{n}),
\eeq
where ${\rm span \, }F$ is the complex subspace of $\mathbb C^{n}$ spanned by the frame
$F$.  If $F, F^{\prime } \in  S(k, \mathbb C^{n})$ define the same subspace, then $F^{\prime } =
FU$ for some $U \in  U(k)$.  Thus we have 
\beq \label{gr:cosets}
G(k, \mathbb C^{n}) \cong  S(n, \mathbb C^{n})/U(k). 
\eeq
In order to parametrize $G(k, \mathbb C^{n})$, we must choose a unique
representative $F \in  S(k, \mathbb C^{n})$ from (\ref{gr:cosets}).
Note that
\beq  \label{gr:piOmegasigma}
G(k, \mathbb C^{n}) \supset  \Omega _{\sigma } = \{ W \in  G(k, \mathbb C^{n}); \det F(W)_{\sigma (n-k+1), \sigma (n-k+2), \ldots , \sigma (n)} \neq 0\}
= \pi _{2}(\tilde{\Omega } _{\sigma }) .
\eeq
For $W \in  \Omega _{\sigma }$, we can choose a unique representative from the coset
$F(W)U(k)$.  In fact, since the submatrix
$Y_{\sigma } = F_{\sigma (n-k+1), \ldots , \sigma (n)}$ is nonsingular, from the uniqueness of the polar decomposition
(see \cite{BB15}) we have
\beq  \label{gr:polar}
Y_{\sigma }^{*} = U\vert Y_{\sigma }^{*}\vert 
\eeq    
for a unique $U \in  U(k)$.  Consequently
$$    Y_{\sigma } = \vert Y_{\sigma }^{*}\vert U^{*} $$ 
for a unique $U^{*} \in  U(k)$.  So, we select a unique representative
$\begin{pmatrix} X_{\sigma }^{\prime } \cr Y_{\sigma }^{\prime } \end{pmatrix}$ which corresponds to $W \in  \Omega _{\sigma }$
such that $Y_{\sigma }^{\prime }$ = $Y_{\sigma }U$ =
$\vert Y_{\sigma }^{*}\vert $ is a positive operator and $X_{\sigma }^{\prime }$ = $X_{\sigma } U$ =
$F_{\sigma (1), \ldots , \sigma (n-k+1)}U$.  Since the column vectors of
$\begin{pmatrix} X_{\sigma }^{\prime } \cr Y_{\sigma }^{\prime } \end{pmatrix}$ gives an orthonormal
frame, we have\newline
$$    X_{\sigma }^{\prime *}X_{\sigma }^{\prime } + Y_{\sigma }^{\prime 2} =  (X_{\sigma }^{\prime *}, Y_{\sigma }^{\prime }) \begin{pmatrix} X_{\sigma }^{\prime } \cr Y_{\sigma }^{\prime }
\end{pmatrix} = I_{k}.$$ 
This shows that $X_{\sigma }^{\prime } \in  B(n-k, k)$ and $Y_{\sigma }^{\prime } = (I_{k} - X_{\sigma }^{\prime *}X_{\sigma }^{\prime })^{1/2}$.
Thus $\Omega _{\sigma }$ is parametrized by $B(n-k, k)$. 

This can be also understood (by showing directly that there is a 
1 to 1 correspondence between $B(n-k, k)$ and $M(n-k, k)$.
For the matrix  $Z$ of (\ref{gr:Z}) introduce
$$    \begin{pmatrix} X \cr Y \end{pmatrix}
= \begin{pmatrix} Z \cr I \end{pmatrix} (Z^{*}Z + I)^{-1/2}$$ 
Then we have
$$    I - X^{*}X = Y^{2} > 0, \  {\rm and \, }\,  0 \leq  X^{*}X < I \  {\rm and \, }\,  Y = (I - X^{*}X)^{1/2}. $$ 
Therefor the mappings
$$    \begin{pmatrix} Z \cr I \end{pmatrix} \rightarrow  \begin{pmatrix} X \cr Y \end{pmatrix}$$ 
and
$$    \begin{pmatrix} X \cr Y \end{pmatrix} \rightarrow  \begin{pmatrix} X \cr Y \end{pmatrix} (I - X^{*}X)^{-1/2} = \begin{pmatrix} Z \cr I \end{pmatrix} $$ 
give a 1-1 onto correspondence, and the mappings
$$    M(n-k, k) \ni  Z \rightarrow  X = Z(Z^{*}Z + I)^{-1/2} \in  B(n-k, k),$$  
$$    B(n-k, k) \ni  X \rightarrow  Z = X(I - X^{*}X)^{-1/2} \in M(n-k, k)$$ 
give a 1-1 onto correspondence between $M(n-k, k)$ and $B(n-k,
k)$.

Let
$\tilde{\psi }_{\sigma }$ be the mapping
$$
\tilde{\psi }_{\sigma }: \tilde{\Omega }_{\sigma } \ni F \rightarrow X_{\sigma }U \in B(n-k, k), \ F_{\sigma } = \begin{pmatrix} X_{\sigma } \cr Y_{\sigma } \end{pmatrix}, \ Y_{\sigma }U = \abs{Y_{\sigma }^*}.
$$
Then $\tilde{\psi }_{\sigma }$ induces the mapping
$$
\psi _{\sigma }: \Omega _{\sigma } \ni \pi _{2}(F) \rightarrow X_{\sigma }U \in B(n-k, k), \ F_{\sigma } = \begin{pmatrix} X_{\sigma } \cr Y_{\sigma } \end{pmatrix}, \ Y_{\sigma }U = \abs{Y_{\sigma }^*}
$$
because $\pi _{2}(F) = \pi _{2}(F^{\prime })$ implies $\psi _{\sigma }(F) = \psi _{\sigma }(F^{\prime })$.

\begin{prop} \label{gr:kappa}
	The mapping 	 $\kappa _{\sigma }: B(n-k, k) \rightarrow  \Omega _{\sigma } = \pi _{2}(\tilde{\Omega }_{\sigma} )$ defined by
	$$    \kappa _{\sigma}: B(n-k, k) \ni  X \rightarrow U_{\sigma } \pi _{2} \begin{pmatrix} X \cr (I_{k} - X^{*}X)^{1/2} \end{pmatrix} \in  \Omega _{\sigma}$$ 
	satisfies $\kappa _{\sigma } \circ \psi _{\sigma } = {\rm id}$ and $\psi _{\sigma } \circ \kappa _{\sigma } = {\rm id}$. 
\end{prop}
\begin{proof}  This is obvious.
\end{proof}

It is easily seen that for any $F, F^{\prime } \in  S(k, \mathbb C^{n})$ there exists $U \in U(n)$ such that $F = UF^{\prime}$, that is $U(n)$ acts transitively on $S(k, \mathbb C^{n})$.  Let
$x = \begin{pmatrix} e_{n-k+1}, & e_{n-k+2}, & \ldots , & e_{n} \end{pmatrix} \in  S(k, \mathbb C^{n})$.
Then the isotropy subgroup of $U(n)$ at $x$ is $U(n-k) \times \{I_k\}$.

Let $F, F^{\prime } \in  S(k, \mathbb C^{n})$, and suppose $\pi _{2}(F) = \pi _{2}(F^{\prime })$.  Then there
exists $Q \in  U(k)$ such that $F = F^{\prime }Q$.  Let $U \in  U(n)$.  Then $UF =
UF^{\prime }Q$, i.e.,  $\pi _{2}(UF) = \pi _{2}(UF^{\prime })$.  This shows that $U(n)$ acts on
$G(k, \mathbb C^{n}) = S(k, \mathbb C^{n})/U(k)$ by
$$      U \pi _{2}(F) = \pi _{2}(UF).$$
Since $U(n)$ acts on $S(k, \mathbb C^{n})$ transitively, $U(n)$ acts on $G(k, \mathbb C^{n}) = \pi _2 (S(k, \mathbb C^{n}))$ transitively.  {This shows that $G(k, \mathbb C^{n})$ is a homogeneous space of $U(n)$ (see \cite{Ma72}).}
Let $y \in  G(k, \mathbb C^{n})$ be a $k$ dimensional subspace of $\mathbb C^{n}$
spanned by the vectors $(e_{n-k+1}, e_{n-k+2}, \ldots , e_{n})$.  The isotropy subgroup of $U(n)$ at 
$y$ is $U(n-k) \times  U(k)$ and thus we have
$$      U(n)/(U(n-k) \times  U(k)) \cong  G(k, \mathbb C^{n}).$$
For
$$    x = \begin{pmatrix} e_{n-k+1}, & e_{n-k+2}, & \ldots , & e_{n} \end{pmatrix} = \begin{pmatrix} O \cr I_{k} \end{pmatrix} \in  S(k, \mathbb C^{n})$$ 
and
$$    g = \begin{pmatrix} W & X \cr V & Y \end{pmatrix} \in  U(n),$$ 
denote by $F = \begin{pmatrix} X \cr Y \end{pmatrix}$  the last $k$ colums of $g$ with   a $k \times  k$ matrix $Y$.
Now define the mapping $\pi _{1}: U(n) \rightarrow  S(k, \mathbb C^{n})$ by
\beq \label{gr:pi1}
\pi _{1}(g) = gx = \begin{pmatrix} W & X \cr V & Y \end{pmatrix}
\begin{pmatrix} O \cr I_{k} \end{pmatrix}
= \begin{pmatrix} X \cr Y \end{pmatrix} \in  S(k, \mathbb C^{n}). 
\eeq

Suppose that $Y$ is a regular $k \times  k$ matrix, i.e., $F \in  \tilde{\Omega } _{e}$ (= $\tilde{\Omega } _{\sigma }$ for $\sigma 
= e$ the identity permutation).  Then there is a unique $Q \in  U(k)$
such that $YQ = \vert Y^{*}\vert $.  Put 
$$    \begin{pmatrix} X^{\prime } \cr Y^{\prime } \end{pmatrix} 
= \begin{pmatrix} XQ \cr \vert Y^{*}\vert  \end{pmatrix}. $$ 
Then we have
$$    \pi _{2} \begin{pmatrix} X \cr Y \end{pmatrix} = \pi _{2} \begin{pmatrix} X^{\prime } \cr Y^{\prime } \end{pmatrix}.$$ 
Let
\beq \label{gr:gWX}
g^{\prime } = W(X^{\prime })
= \begin{pmatrix} (I - X^{\prime }X^{\prime *})^{1/2} & X^{\prime } \cr -X^{\prime *} & Y^{\prime } 
\end{pmatrix}. 
\eeq   
Then $g^{\prime } \in U(n)$ by Proposition \ref{li:WX} and we have
$$    \pi _{1}(g^{\prime }) = \begin{pmatrix} (I - X^{\prime }X^{\prime *})^{1/2} & X^{\prime } \cr -X^{\prime *} & Y^{\prime }
\end{pmatrix} x 
= \begin{pmatrix} X^{\prime } \cr Y^{\prime } \end{pmatrix} .$$ 

For 
\beq  \label{gr:pi}
\pi = \pi _{2}\circ \pi _{1}.
\eeq 
we have
$$      \pi (g) = \pi (g^{\prime }).$$ 
\begin{defn}
	A mappping $\iota : G(k,\mathbb C^{n}) \rightarrow U(n)$ is a section of $U(n)$ on $G(k,\mathbb C^{n})$ for $\pi : U(n) \rightarrow G(k,\mathbb C^{n})$ if it satisfies
	\beq \label{gr:section}
	\pi (\iota (x)) = x  \quad \textrm{for all}\; x \in G(k,\C^n). 
	\eeq
	If $\iota $ is defined only on a subset $\Omega \subset G(k,\mathbb C^{n})$ and satisfies (\ref{gr:section}) there, $\iota $ is called a local section. 
\end{defn}
Thus, by (\ref{gr:gWX}), we have constructed a local section:
\begin{prop} \label{gr:lambda}
	Let
	\beq  \label{gr:WX}
	W(X) = \begin{pmatrix} (I_{n-k} - XX^{*})^{1/2} & X \cr -X^{*} & (I_{k} -
		X^{*}X)^{1/2} \end{pmatrix} 
	\eeq
	for $X \in B(n-k, k)$.
	Then the mapping
	\beq \label{gr:local section} 
	\iota _{e} = W \circ \psi _{e}: \Omega _{e} \ni  \pi (g) \rightarrow  g^{\prime } \in  U(n).
	\eeq  	
	gives a local section of $U(n)$ on $\Omega _{e}$ for $\pi : U(n) \rightarrow G(k, \C^n)$.
\end{prop}
\begin{proof} In Proposition \ref{li:WX} it is shown in detail that the matrix $W(X)$ in (\ref{gr:WX}) is unitary. The proof of the remaining part of the statement is straight forward.	
\end{proof}

From now on, we use the notation $\tilde{\Omega } $ for $\tilde{\Omega } _{e}$, and $\Omega $ for
$\Omega _{e}$, where $e$ is the identity permutation.
\begin{prop} \label{gr:Omegasigma}
	\beq \label{gr:Usigma}
	\tilde{\Omega } _{\sigma } = U_{\sigma }\tilde{\Omega }  , \  {\rm and \, }\,  \Omega _{\sigma } = U_{\sigma }\Omega .
	\eeq
\end{prop}
\begin{proof}  Let $F \in  \tilde{\Omega } _{\sigma }$.  Then $U_{\sigma }^{-1}F = F_{\sigma } = \begin{pmatrix} X_{\sigma } \cr Y_{\sigma }
	\end{pmatrix}$ belongs to $\tilde{\Omega } $ since $Y_{\sigma }$ is regular.  Let $F \in  \tilde{\Omega } $
	and $F^{\prime } = U_{\sigma }F$.  Then $F^{\prime } \in  \tilde{\Omega } _{\sigma }$.  In fact, $F^{\prime }_{\sigma } = U_{\sigma }^{-1}F^{\prime } = F$ and $Y^{\prime }_{\sigma }
	= Y$ is regular.  Thus we have $\tilde{\Omega } _{\sigma } = U_{\sigma }\tilde{\Omega } $, and by (\ref{gr:piOmegasigma})
	$$      \Omega _{\sigma } = \pi _{2}(\tilde{\Omega } _{\sigma }) = \pi _{2}(U_{\sigma }\tilde{\Omega } ) = U_{\sigma } \pi _{2}(\tilde{\Omega } ) = U_{\sigma }\Omega .$$
\end{proof}
In Section \ref{liealg} it is explained why we consider $W(X)$ of (\ref{gr:WX}).
Now we extend Proposition \ref{gr:lambda} to $\Omega _{\sigma }$.
\begin{prop} \label{gr:lambdasigma}
	The mapping\newline
	$$    W_{\sigma } = U_{\sigma } W : B(n-k, k) \ni  X \rightarrow  U_{\sigma }W(X)$$ 
	satisfies $\psi _{\sigma } \circ \pi \circ W_{\sigma }(X) = X$.
\end{prop}
\begin{proof}  For $X \in  B(n-k, k)$  we find
	$$    \pi _{1}\circ W _{\sigma }(X) = \pi _{1}(U_{\sigma }W(X)) = U_{\sigma } \begin{pmatrix} X \cr (I_{k} -
	X^{*}X)^{1/2} \end{pmatrix} = U_{\sigma }F = F^{\prime },$$ 
	and $\psi _{\sigma }\circ \pi _{2}(F^{\prime }) = \psi _{e} U_{\sigma }^{-1}\{ F^{\prime }U; U\in U(k)\}  = \psi _{e} \{ FU; U\in U(k)\}  = X$.  Thus we
	have $\psi _{\sigma } \circ \pi \circ W _{\sigma }(X) = \psi _{\sigma } \circ \pi _{2}\circ \pi _{1}\circ W _{\sigma }(X) = X$. \end{proof}
\begin{cor}
	The mapping 
	$$      \iota _{\sigma } = U_{\sigma } W \circ \psi _{\sigma }: \Omega _{\sigma } \ni x  \rightarrow  W(\psi _{\sigma }(x)) \in  U(n)$$
	gives a local section of $U(n)$ on $\Omega _{\sigma } \subset  G(k, \mathbb C^{n})$ for $\pi : U(n) \rightarrow  G(k, \mathbb C^{n})$.
\end{cor}
\begin{proof}  Proposition \ref{gr:kappa}  shows that $\psi _{\sigma }$ is bijective; by  Proposition \ref{gr:lambdasigma} we conclude. \end{proof}
\vskip 12pt \noindent
Let $S_{n}$ be the set of all permutations of $\{ 1, 2, \ldots , n\} $, and
define an order on $S_{n}$ as follows (lexicographic ordering).  Let
$\sigma , \sigma ^{\prime } \in  S_{n}$ then $\sigma  < \sigma ^{\prime }$ if there exists $s \in  \{ 1, 2, \ldots , n\} $
such that $\sigma (j) = \sigma ^{\prime }(j)$ $(j = 1, \ldots , s-1)$ and $\sigma (s) < \sigma ^{\prime }(s)$.
Then $S_{k, n} \subset  S_{n}$ is a well-ordered set, and $S_{k, n} = \{ \sigma _{1} < \ldots <
\sigma _{m}\} $ for $m = {n \choose k}$.
\vskip 12pt \noindent
For 
\beq  \label{gr:Vj}
V_{j} = \Omega _{\sigma _{j}} \backslash  \cup _{i=1}^{j-1} \Omega _{\sigma _{i}}
\eeq
one finds that
$$      G(k, \mathbb C^{n}) = \cup _{j=1}^{m} V_{j}, \  m = {n \choose k}$$
is a disjoint union.  We can construct a section $\iota : G(k, \mathbb C^{n}) \rightarrow  
U(n)$ for $\pi : U(n) \rightarrow  G(k, \mathbb C^{n})$ as follows.  
\begin{defn}
	Let $x \in  G(k, \mathbb C^{n})$.  Define the section $\iota (x)$ by
	$$      \iota (x) = \iota _{\sigma _{j}}(x) \  {\rm if \, }\,  x \in  V_{j}.$$
\end{defn}

\begin{prop}  \label{gr:uniqueexpression}
	Let $\iota $ be a section of $U(n)$ on $G(k, \mathbb C^{n})$ for $\pi $.  Then for any $g \in  U(n)$, there is a unique $h$ such that   
	$$      g = \iota (\pi (g)) h, \  h \in  U(n-k)\times U(k) .$$
\end{prop}
Proof.  Let $g^{\prime } = \iota (\pi (g))$.  Since $\pi (g) = \pi (g^{\prime })$ and
$$      \pi (g) = \{ gh; h \in  U(n-k)\times U(k)\} ,$$
there exists $h \in  U(n-k)\times U(k)$ such that $g = g^{\prime }h$.  If $g$ has two
expression $g = g^{\prime }h = g^{\prime }h^{\prime }$.  Then we have $h = g^{\prime -1}g = h^{\prime }$. $\Box $

Now we come to the parametrization of unitary matrices by the canonical coset decomposition.  
Symbolically the canonical coset decomposition is:
$$      U(n) = U(n)/(U(n-k_m) \times  U(k_m)) \cdot  U(n-k_m)/(U(n-k_{m-1} - k_{m}) \times  U(k_{m-1}))$$
$$      \cdot  \cdots \cdot  U(k_{1} + k_{2})/(U(k_{1}) \times  U(k_{2})) \cdot  (U(k_1) \times  \cdots \times  U(k_m)).$$
The following Proposition shows the precise meaning of the above formula.
\begin{prop}  \label{gr:canonicaldec}
	For any section $\iota _{j}$ of $U(n - k_{j+1} - \cdots - k_{m})$
	on $G(k_{j}, \mathbb C^{n-k_{j+1}- \cdots - k_{m}})$ for $\pi _{j}: U(n - k_{j+1} - \cdots -k_{m})
	\rightarrow  G(k_{j}, \mathbb C^{n-k_{j+1}- \cdots - k_{m}})$ there is a unique surjection 
	$$  f: U(n) \ni g \rightarrow  (z_{m}, z_{m-1}, \ldots , z_{2}) \in  G(k_{m},  \mathbb C^{n}) \times  G(k_{m-1}, \mathbb C^{n-k_{m}}) \times $$
	$$   \cdots \times  G(k_{2}, \mathbb C^{n-k_{3}- \cdots - k_{m}})$$
	and a unique $h \in  U(k_{1}) \times  \cdots \times 
	U(k_{m})$ such that
	$$  g = \iota _m (z_m) \begin{pmatrix} \iota _{m-1}(z_{m-1}) & 0 \cr 0 & I_{k_m} \end{pmatrix} 
	\cdots \begin{pmatrix} \iota _{2}(z_{2}) & 0 \cr 0 & I_{k_3 + \ldots + k_{m}}
	\end{pmatrix} h
	. $$ 
\end{prop}
\begin{proof} According to Proposition \ref{gr:uniqueexpression} there exists a unique $H_{m} = (g_{m}, h_{m})
	\in  U(n-k_{m}) \times  U(k_{m})$ such that 
	$$   g = \iota _{m}(z_{m})H_{m} = \iota _{m}(z_{m})(g_{m}, h_{m}) = \iota _{m}(z_{m})
	\begin{pmatrix} g_{m} & 0 \cr 0 & h_{m} \end{pmatrix},  $$ 
	where $z_{m} = \pi _{m}(g)$. \newpage
	In the same way, we have a unique element $H_{m-1} = (g_{m-1}, h_{m-1}) \in  U(n-k_{m}-k_{m-1}) \times  U(k_{m-1})$
	such that 
	$$   g_{m} = \iota _{m-1}(z_{m-1})H_{m-1} = \iota _{m-1}(z_{m-1})(g_{m-1}, h_{m-1}) = \iota _{m-1}(z_{m-1})
	\begin{pmatrix} g_{m-1} & 0 \cr 0 & h_{m-1} \end{pmatrix},  $$ 
	for $z_{m-1} = \pi _{m-1}(g_{m})$, and
	$$    g = \iota _m (z_m) \begin{pmatrix} \iota _{m-1}(z_{m-1}) \begin{pmatrix} g_{m-1} & 0 \cr 0 & h _{m-1} \end{pmatrix} 
	& 0 \cr 0 & h_{m} \end{pmatrix} $$ 
	$$      = \iota _m (z_m) \begin{pmatrix} \iota _{m-1}(z_{m-1})
	& 0 \cr 0 & I_{k_m} \end{pmatrix}
	\begin{pmatrix} g_{m-1} & 0 & 0 \cr 0 & h _{m-1} & 0 \cr 0 & 0 & h _{m} \end{pmatrix} . $$ 
	Continuing this procedure, we arrive at
	$$  g = \iota _m (z_m) \begin{pmatrix} \iota _{m-1}(z_{m-1}) & 0 \cr 0 & I_{k_m} \end{pmatrix} 
	\cdots \begin{pmatrix} \iota _{2}(z_{2}) & 0 \cr 0 & I_{k_3 + \ldots + k_{m}}
	\end{pmatrix} 
	\begin{pmatrix} h_{1} & \cdots & 0 \cr \vdots & \ddots & \vdots \cr 0 & \cdots & h _{m} \end{pmatrix} . $$ 
	The relation
	$$   f(g^{\prime }) = (z_m, z_{m-1}, \ldots , z_2)$$
	for
	$$  (z_{m}, z_{m-1}, \ldots , z_{2}) \in  G(k_{m},  \mathbb C^{n}) \times  G(k_{m-1}, \mathbb C^{n-k_{m}}) \times \ldots \times  G(k_{2}, \mathbb C^{n-k_{3}- \ldots - k_{m}})$$
	and
	$$  g^{\prime } = \iota _m (z_m) \begin{pmatrix} \iota _{m-1}(z_{m-1}) & 0 \cr 0 & I_{k_m} \end{pmatrix} 
	\cdots \begin{pmatrix} \iota _{2}(z_{2}) & 0 \cr 0 & I_{k_3 + \ldots + k_{m}}
	\end{pmatrix}$$
	shows the surjectivity of $f$. 
\end{proof}

Degenerate density matrices with a diagonal matrix of eigenvalues of the form
$$
D_n(\bs{\lambda}) = {\rm diag }_n(
\lambda_1 I_{k_1},\lambda_2 I_{k_2},\ldots,
\lambda_m I_{k_m} ),  \ I_{k_j}: \  k_j \times k_j {\rm \ identity \ matrix}.
$$
are parametrized by
\beq \label{gr:degenerateden}
\Lambda_m^{\neq} \times G(k_{m},  \mathbb C^{n}) \times  G(k_{m-1}, \mathbb C^{n-k_{m}}) \times \cdots \times  G(k_{2}, \mathbb C^{n-k_{3}- \cdots - k_{m}}).
\eeq

We identify $U(k_{1} + \cdots + k_{j-1})$ and $U(k_{1} + \cdots + k_{j-1}) \times \{I_{k_j + \cdots + k_m}\}$ in $U(n)$.
Define the equivalence relation $\sim$ in $U(n)$ by
$$      g \sim  g^{\prime } \Leftrightarrow  g^{-1}g^{\prime } \in  U(k_1)\times \cdots
\times	U(k_m).$$
Accordingly the elements of 
$$    U(n)/(U(k_1) \times \cdots \times	U(k_m)) = U(n)/\sim  $$ 
are the equivalence classes
$$      [g] = \{ gv; v \in  U(k_1)\times \cdots \times	U(k_m)\} $$
\begin{prop}
	For $g_j \in U(k_{1} + \cdots + k_{j-1} + k_{j})$  there is a $g_{j-1} \in U(k_{1} + \cdots + k_{j-1})$ such that
	$$    [g_j] = [\iota _j(\pi _j(g_j))g_{j-1}] ,$$
	where $\pi _j(g_j)$ and $[g_{j-1}]$ are uniquely determined by $[g_j]$, for any section $\iota _j$ of $U(k_{1} + \cdots + k_{j-1} + k_{j})$ on $U(k_{1} + \cdots + k_{j-1} + k_{j})/(U(k_{1} + \cdots + k_{j-1}) \times U(k_{j}))$.
\end{prop}
\begin{proof}
	Since $(U(k_1) \times \cdots \times	U(k_{j-1}) \times U(k_j)) \subset U(k_{1} + \cdots + k_{j-1}) \times U(k_{j})$, $[g_j] = [g_j^{\prime }]$ implies $\pi _j(g_j) = \pi _j(g_j^{\prime })$ for $g_j, g_j^{\prime } \in U(k_{1} + \cdots + k_{j-1} + k_{j})$.  
	
	It follows from Proposition \ref{gr:uniqueexpression} that
	$g_j = \iota (\pi (g_j)) H$ holds for $H = (g_{j-1}, h) \in  U(k_{1} + \cdots + k_{j-1})\times U(k_j)$.  Thus we have
	$[g_j] = [\iota (\pi (g_j)) H] = [\iota (\pi (g_j)) g_{j-1}]$.  
	
	Let $[g_j] = [g_j^{\prime }]$ then there is a $g_{j-1}^{\prime } \in U(k_{1} + \cdots + k_{j-1})$ such that $[g^{\prime }] = [\iota _j(\pi _j(g))g_{j-1}^{\prime }]$.  Therefore there exists $v \in (U(k_1) \times \cdots \times U(k_{j-1}) \times U(k_j))$ such that $\iota _j(\pi _j(g_j))g_{j-1}^{\prime } = \iota (\pi (g_j)) g_{j-1} v$, and $g_{j-1}^{\prime } = g_{j-1} v$.  This shows $[g_{j-1}] = [g_{j-1}^{\prime }]$.
\end{proof}
In the same way as  Proposition \ref{gr:canonicaldec} one proves the following result.
\begin{prop}  
	For any section $\iota _{j}$ of $U(n - k_{j+1} - \cdots - k_{m})$
	on $G(k_{j}, \mathbb C^{n-k_{j+1}- \cdots - k_{m}})$ for $\pi _{j}: U(n - k_{j+1} - \cdots -k_{m})
	\rightarrow  G(k_{j}, \mathbb C^{n-k_{j+1}- \cdots - k_{m}})$,
	there is a unique mapping
	$$  \phi : U(n)/(U(k_1) \times \cdots \times U(k_m)) \ni [g] \rightarrow  $$
	$$  (z_{m}, z_{m-1}, \ldots , z_{2}) \in  G(k_{m},  \mathbb C^{n}) \times  G(k_{m-1}, \mathbb C^{n-k_{m}}) \times 
	\cdots \times  G(k_{2}, \mathbb C^{n-k_{3}- \cdots - k_{m}})$$
	such that
	$$  [g] = [g^{\prime }] \ {\rm for} \ g^{\prime } = \psi (\phi ([g])),$$
	where
	$$  \psi (z_{m}, z_{m-1}, \ldots , z_{2}) = \iota _m (z_m) \begin{pmatrix} \iota _{m-1}(z_{m-1}) & 0 \cr 0 & I_{k_m} \end{pmatrix} 
	\cdots \begin{pmatrix} \iota _{2}(z_{2}) & 0 \cr 0 & I_{k_3 + \ldots + k_{m}}
	\end{pmatrix} 
	. $$ 
\end{prop}
Let 
$$  \pi : U(n) \ni g \rightarrow [g] \in U(n)/(U(k_1) \times \cdots \times U(k_m)). $$
The above proposition shows that the mapping $\pi \circ \psi $ is surjective.  Since $\iota _{j}$ is a section of $U(n - k_{j+1} - \cdots - k_{m})$ on $U(n - k_{j+1} - \cdots - k_{m})/(U(n - k_{j} - \cdots - k_{m}) \times U(k_j)) $, and 
$U(k_1) \times \cdots \times U(k_m) \subset U(n - k_{j} - \cdots - k_{m}) \times U(k_j)$, the mapping $\pi \circ \psi $ is injective.  Consequently the mappings $\phi $ and $\pi \circ \psi $ are inverses to each other, and thus we get
\begin{prop}  \label{gr:flagsection}
	There is a bijection $\phi = (\pi \circ \psi )^{-1}$ between the flag manifold $U(n)/(U(k_1) \times \cdots \times U(k_m))$ and the direcat product $G(k_{m},  \mathbb C^{n}) \times  G(k_{m-1}, \mathbb C^{n-k_{m}}) \times 
	\cdots \times  G(k_{2}, \mathbb C^{n-k_{3}- \cdots - k_{m}})$ of Grassmann manifolds.  The mapping $\psi \circ \phi $ is a section of $U(n)$ on $U(n)/(U(k_1) \times \cdots \times U(k_m))$ with respect to $\pi $.
\end{prop}

\section{Projective space}  \label{projectsp}

In this section, we summarize the results in Section \ref{grassmann} for $k_j=1$.  The Grassmann manifold $G(1, \mathbb C^{n})$, the set of all complex
$1$-dimensional subspaces of $\mathbb C^{n}$, is called the projective space and denoted by $CP^{n-1}$.

$$    CP^{n-1} = G(1, \mathbb C^{n}) = \{z \in \mathbb C^{n}; z \neq 0 \} /(\mathbb C \backslash \{0\}).$$
$S_{1, n}$ is a set of permutation $\sigma $ such that $1 \leq  \sigma (1) < \sigma (2) <
\cdots < \sigma (n-1) \leq  n$ and $1 \leq \sigma (n) \leq n$.  
\begin{rem}
	Let $\sigma (n) = j$.  Then we have
	$\sigma (k) = k$ if $1 \leq k \leq j-1$ and $\sigma (k) = k+1$ if $j \leq k \leq n-1$.
	Thus $\sigma \in S_{1, n}$ is characterized by $j = \sigma (n) \in \{1, 2, \ldots , n\}$.
\end{rem}
Let $j = \sigma (n)$ and $z_j = z_{\sigma (n)}$ is the $j$-th component of $z \in \mathbb C^{n}$.
Define
$$    \Omega _{j } = \{ z \in \mathbb C^{n}; z_j  \neq 0\}/(\mathbb C \backslash \{0\}) \subset CP^{n-1}.$$
Since $z \neq 0$, there is a $j \in \{1, 2, \ldots , n\}$ such that $z_j \neq 0$.  Thus
we have
$$    CP^{n-1} = \cup _{j \in \{1, 2, \ldots , n\}} \Omega _{j }.$$

Define the mapping
\beq \label{gr:chart}
\phi _{j } = \phi _{\sigma }: \Omega _{\sigma } = \Omega _{j } \ni  w \rightarrow  (z_{\sigma (1)}, \ldots , z_{  \sigma (n-1)})/z_{\sigma (n)}
= (z_1 \ldots z_{j-1}, z_{j+1}, \ldots , z_n)/z_j
\in  \mathbb C^{n-1}, 
\eeq     
where $\sigma (n) = j$.
This map $\phi _{j }$ gives the homeomorphism
$$      \Omega _{j } \cong  \mathbb C^{(n-1)}$$
and the set $\{ (\Omega _{j }, \phi _{j }); j  \in  \{1, 2, \ldots , n\}\} $  an atlas of $CP^{n-1}$.

There is another parametrization of $\Omega _{j }$ which is more
convenient for our purpose.

Introduce
$$      B(n-1) = \{ x \in  \mathbb C^{(n-1)}; x^{*}x < 1\};$$
the set $\Omega _{j }$ can be parametrized by the set $B(n-1)$.  We will show this in the following.

Note that $S(1, \mathbb C^{n}) = S(\mathbb C^{n}) = \{z \in \mathbb C^{n}; \norm{z} = 1\}$. 
With 
$$    \tilde{\Omega } _{j } = \{ z \in  S(\mathbb C^{n}); z_j \neq 0\}  \subset  S(\mathbb C^{n}).$$ 
one has
$$      S(\mathbb C^{n}) = \cup _{j \in \{1, 2, \ldots , n\}} \tilde{\Omega } _{j }.$$
\begin{rem}  \label{pr:usigma}
	Let $\sigma \in S_{1, n}$ such that $\sigma (n) = j$ and denote $U_{\sigma }$ by $U_j$.  Then we have
	$U_{j }e_{k} = e_{k}$ if $1 \leq k \leq j-1$, $U_{j }e_{k} = e_{k+1}$ if $j \leq k \leq n-1$ and $U_j e_n = e_j$.  
\end{rem}
Let $\pi _{2}$ be the surjective mapping
$$      \pi _{2}: S(\mathbb C^{n}) \ni  z \rightarrow  w = {\rm span \, }z \in  CP^{n-1},$$
where ${\rm span \, }z$ is the complex line spanned by $z$.  If $z, z^{\prime } \in  S(\mathbb C^{n})$ define the same line, then $z^{\prime } =
z e^{i\theta }$ for some $e^{i\theta } \in  U(1)$.  Thus we have 
\beq \label{pr:cosets}
CP^{n-1} \cong  S(\mathbb C^{n})/U(1). 
\eeq
In order to parametrize $CP^{n-1}$, we must choose a unique
representative $z \in  S(\mathbb C^{n})$ from (\ref{pr:cosets}).
Note that
$$      CP^{n-1} \supset  \Omega _{j } = \{ z \in  S(\mathbb C^{n}); z_j \neq 0\}/U(1)
= \pi _{2}(\tilde{\Omega } _{j }) .$$

Let $\tilde{\psi }_j$ be the mapping
$$
\tilde{\psi }_j: \tilde{\Omega }_j \ni z \rightarrow (z_1 \ldots z_{j-1}, z_{j+1} \ldots , z_n)^T/e^{i\theta } \in B(n-1), \ z_j = \abs{z_j} e^{i\theta }.
$$
Then $\tilde{\psi }_j$ induces the mapping
$$
\psi _j: \Omega _j = \pi _2(\tilde{\Omega }_j) \ni \pi _2(z) \rightarrow (z_1 \ldots z_{j-1}, z_{j+1} \ldots , z_n)^T/e^{i\theta } \in B(n-1), \ z_j = \abs{z_j} e^{i\theta }
$$
because $\pi _2(z) = \pi _2(z^{\prime })$ implies $\tilde{\psi }_j(z) = \tilde{\psi }_j(z^{\prime})$.
\begin{prop} \label{pr:kappa}
	The mapping 
	$$    \kappa _j: B(n-1) \ni  x \rightarrow  \pi _{2}( 
	(z_1 \ldots z_{j-1}, (1 - x^{*}x)^{1/2}, z_{j} \ldots , z_{n-1}))^T
	\in  \Omega _j.$$ 
	satisfies $\psi _j \circ \kappa _j = {\rm id}$ and $\kappa _j \circ \psi _j = {\rm id}$.	
\end{prop}

Let $x = e_{n}$.  Since $U(n)$ acts on $S(\mathbb C^{n})$ transitively, and
the isotropy group of $x$ is $U(n-1) \times \{1\}$,
$U(n)$ acts on $CP^{n-1}= S(\mathbb C^{n})/U(1)$ transitively by
$$      U \pi _{2}(z) = \pi _{2}(Uz).$$
Let $y \in  CP^{n-1}$ be a complex line of $\mathbb C^{n}$
spanned by the veactor $e_{n}$.  The isotropy group of 
$y$ is $U(n-1) \times  U(1)$, and we have
$$      U(n)/(U(n-1) \times  U(1)) \cong  CP^{n-1}.$$
Let
$$    e = e_{n} = \begin{pmatrix} 0 \cr \vdots \cr 0 \cr 1 \end{pmatrix} = \begin{pmatrix} O \cr 1 \end{pmatrix} \in  S(\mathbb C^{n})$$ 
and
$$    g = \begin{pmatrix} W & x \cr v & y \end{pmatrix} \in  U(n),$$ 
where $z = \begin{pmatrix} x \cr y \end{pmatrix}$ is the last colum
of $g$ and $y \in \mathbb C^{1}$.

Define the mapping $\pi _{1}: U(n) \rightarrow  S(\mathbb C^{n})$ by
$$    \pi _{1}(g) = ge = \begin{pmatrix} W & x \cr v & y \end{pmatrix}
\begin{pmatrix} O \cr 1 \end{pmatrix}
= \begin{pmatrix} x \cr y \end{pmatrix} \in  S(\mathbb C^{n}). $$ 
Suppose that $y \neq 0$.  Then there is a unique $e^{i \theta } \in  U(1)$
such that $ye^{i \theta } = \vert y\vert $.  For 
$$    \begin{pmatrix} x^{\prime } \cr y^{\prime } \end{pmatrix} 
= \begin{pmatrix} xe^{i \theta } \cr \vert y\vert  \end{pmatrix}. $$ 
we have\newline
$$    \pi _{2} \begin{pmatrix} x \cr y \end{pmatrix} = \pi _{2} \begin{pmatrix} x^{\prime } \cr y^{\prime } \end{pmatrix}.$$ 
And with
\beq \label{pr:gWx}
g^{\prime } = W(x^{\prime })
= \begin{pmatrix} (I_{n-1} - x^{\prime }x^{\prime *})^{1/2} & x^{\prime } \cr -x^{\prime *} & y^{\prime }
\end{pmatrix}. 
\eeq   
we get 
$$    \pi _{1}(g^{\prime }) = \begin{pmatrix} (I_{n-1} - x^{\prime }x^{\prime *})^{1/2} & x^{\prime } \cr -x^{\prime *} & y^{\prime }
\end{pmatrix} e  
= \begin{pmatrix} x^{\prime } \cr y^{\prime } \end{pmatrix} .$$ 
Finally introduce $\pi  = \pi _{2}\circ \pi _{1}$ and $\Omega _{\sigma } = \pi _{2}(\tilde{\Omega } _{\sigma })$.  Then we have
$$      \pi (g) = \pi (g^{\prime }).$$
Thus we have constructed a local section
\beq \label{pr:local section} 
\iota : \Omega _{n} \ni  \pi (g) \rightarrow  g^{\prime } \in  U(n) 
\eeq   
by (\ref{pr:gWx}).\newline
\begin{prop} \label{pr:lambda}
	For $x \in B(n-1)$ define
	\beq  \label{pr:Wofx}  
	W(x) = \begin{pmatrix} (I_{n-1} - xx^{*})^{1/2} & x \cr -x^{*} & (1 -
		x^{*}x)^{1/2} \end{pmatrix} . 
	\eeq
	Then the mapping
	$$    \iota _{n} = W \circ \psi _{n}: \Omega _{n} \rightarrow   U(n)$$ 
	gives a local section of $U(n)$ on $CP^{n-1}$ for $\pi : U(n) \rightarrow CP^{n-1}$.
\end{prop}
\begin{proof}
	The matrix $W(x)$ in (\ref{pr:Wofx}) is unitary according to Proposition \ref{li:WX}; the proof of the remaining part is obvious by the above preparations.
\end{proof}	

From now on, we use the notation $\tilde{\Omega } $ for $\tilde{\Omega } _{n}$, and $\Omega $ for
$\Omega _{n}$.  Propositions \ref{pr:Omegasigma}, \ref{pr:kappa}, \ref{pr:lambda} are the special cases ($k = 1$) of Propositions \ref{gr:Omegasigma}, \ref{gr:kappa}, \ref{gr:lambda} respectively, and we omit the proofs.
\begin{prop} \label{pr:Omegasigma}
	\beq \label{pr:Usigma}
	\tilde{\Omega } _{j } = U_{j }\tilde{\Omega }  , \  {\rm and \, }\,  \Omega _{j } = U_{j }\Omega .
	\eeq
\end{prop}

Next we extend {Proposition} \ref{pr:lambda} to $\Omega _{j }$.

\begin{prop}  \label{pr:localsection}
	The mapping 
	$$      \iota _{j} = U_{j} W \circ \psi _{j}: \Omega _{j} \ni x  \rightarrow  W(\psi _{j}(x)) \in  U(n)$$
	gives a local section of $U(n)$ on $\Omega _{j} \subset  CP^{n-1}$ for $\pi : U(n) \rightarrow  CP^{n-1}$.
\end{prop}
For $j=1,\ldots,n-1$ introduce 
\beq  \label{pr:Vj}
V_{j} = \Omega _{j} \backslash  \cup _{i=j+1}^{n} \Omega _{i}.
\eeq  
Then
$$      CP^{n-1} = \cup _{j=1}^{n} V_{j}$$
is a disjoint union.  We can construct a section $\iota : CP^{n-1} \rightarrow  
U(n)$ for $\pi : U(n) \rightarrow  CP^{n-1}$ as follows.  
\begin{defn}
	Let $x \in  CP^{n-1}$.  Define the section $\iota (x)$ by
	\beq \label{pr:section}
	\iota (x) = \iota _{j}(x) \  {\rm if \, }\,  x \in  V_{j}.
	\eeq
\end{defn}
\begin{rem}
	Let $\sigma , \sigma ^{\prime } \in S_{1, n}$ be such that $\sigma (n) < \sigma ^{\prime } (n)$.  Then $\sigma > \sigma ^{\prime }$ according to the lexicographic ordering.  This causes the difference between (\ref{pr:Vj}) and (\ref{gr:Vj}) 
\end{rem}
\begin{prop}  \label{pr:uniqueexpression}
	Let $\iota $ be a section of $U(n)$ on $CP^{n-1}$ for $\pi $.  Then for any $g \in  U(n)$, there is a unique $h\in  U(n-1)\times U(1)$ such that   
	$$      g = \iota (\pi (g)) h\   . $$
\end{prop}

Now we are well prepared to present the parametrization of unitary matrices by the canonical coset decompositon.  
Symbolically, the canonical coset decomposition is:
$$      U(n) = U(n)/(U(n-1) \times  U(1)) \cdot  U(n-1)/(U(n-2) \times  U(1))$$
$$      \cdot  \cdots \cdot  U(2)/(U(1) \times  U(1)) \cdot  (U(1) \times  \cdots \times  U(1)).$$
The following proposition shows the precise meaning of the above formula.
\begin{prop}  \label{pr:canonicaldec}
	For any section $\iota _{j}$ of $U(j)$
	on $CP^{j-1}$ for $\pi _{j}: U(j)
	\rightarrow  CP^{j-1}$, there is a unique surjection $f: U(n) \ni g \rightarrow (z_{n}, z_{n-1}, \ldots , z_{2}) \in  CP^{n-1}  \times  CP^{n-2} \times 
	\cdots \times  CP^{1}$ and a unique $h \in  U(1) \times  \cdots \times 
	U(1)$ such that	
	$$  g = \iota _n (z_n) \begin{pmatrix} \iota _{n-1}(z_{n-1}) & 0 \cr 0 & I_{1} \end{pmatrix} 
	\cdots \begin{pmatrix} \iota _{2}(z_{2}) & 0 \cr 0 & I_{n-2}
	\end{pmatrix} h
	. $$ 
\end{prop}
Nondegenerate density matrices are parametrized by
$$    \Lambda_n^{\neq} \times CP^{n-1}  \times  CP^{n-2} \times \cdots \times  CP^{1}.$$

Let $\pi $ be the natural map
$$\pi : U(n) \rightarrow U(n)/(U(1)\times \cdots \times U(1)).$$

\begin{cor}  \label{pr:canonicaldec}
	For any section $\iota _{j}$ of $U(j)$
	on $CP^{j-1}$ for $\pi _{j}: U(j)
	\rightarrow  CP^{j-1}$, there is a unique bijection 
	$$   \phi : U(n)/(U(1)\times \cdots \times U(1)) \ni \pi (g) \rightarrow $$
	$$  (z_{n}, z_{n-1}, \ldots , z_{2}) \in  CP^{n-1}  \times  CP^{n-2} \times 
	\cdots \times  CP^{1}$$
	such that 
	$$ \pi (g) = \pi (g^{\prime }) {\rm \ for \ } g^{\prime } = \psi (\phi (\pi (g))),$$
	where
	$$  \psi (z_n, z_{n-1}, \ldots , z_2) = \iota _n (z_n) \begin{pmatrix} \iota _{n-1}(z_{n-1}) & 0 \\ 0 & I_{1} \end{pmatrix} 
	\cdots \begin{pmatrix} \iota _{2}(z_{2}) & 0 \\ 0 & I_{n-2}
	\end{pmatrix}
	. $$ 
	$\psi \circ \phi $ is a section of $U(n)$ on $U(n)/(U(1)\times \cdots \times U(1))$ for $\pi $.
\end{cor}

\section{Local charts of Grassmannian} \label{jarlskog}
For $F \in  S(k, \mathbb C^{n})$ the submatrix $M(F)_{n-k+1, \ldots , n}$ is not necessarily nonsingular, unles $F \in \tilde{\Omega }$.  From the uniqueness of the polar decomposition
(see \cite{BB15}) we get
\beq  \label{bd:polar}
M(F)_{n-k+1, \ldots , n}^{*} = V\vert M(F)_{n-k+1, \ldots , n}^{*}\vert  
\eeq
for a unique partial isometry $V$ with null space $N(V) = N(M(F)_{n-k+1, \ldots , n}^{*})$ (cp. (\ref{gr:polar})).
There exists $U \in  U(k)$ such that the restriction $U_{\vert N(V)^{\bot }}$ of $U$ to the orthogonal complement $N(V)^{\bot }$ of $N(V)$ is $V$.  Consequently we have
$$    M(F)_{n-k+1, \ldots , n} = \vert M(F)_{n-k+1, \ldots , n}^{*}\vert U^{*} $$ 
for some $U^{*} \in  U(k)$.  So, we can select some frame
$F^{\prime } = \begin{pmatrix} X^{\prime } \cr Y^{\prime } \end{pmatrix} \in  S(k, \mathbb C^{n})$ 
such that $Y^{\prime }$ = $M(F)_{n-k+1, \ldots , n}U$ =
$\vert M(F)_{n-k+1, \ldots , n}^{*}\vert $ is a nonnegative operator and $X^{\prime }$ =
$M(F)_{1, \ldots , n-k}U$.  
Since the column vectors of
$\begin{pmatrix} X^{\prime } \cr Y^{\prime } \end{pmatrix} = \begin{pmatrix} XU \cr YU \end{pmatrix}$ give an orthonormal
frame, we have\newline
$$    X^{\prime *}X + Y^{\prime 2} =  (X^{\prime *}, Y^{\prime }) \begin{pmatrix} X^{\prime } \cr Y^{\prime }
\end{pmatrix} = I_{k}.$$ 
This shows that $X^{\prime } \in  \bar{B}(n-k, k)$ and $Y^{\prime } = (I_{k} - X^{\prime *}X^{\prime })^{1/2}$ where we use the notation
$$      \bar{B}(n-k, k) = \{ X \in  M(n-k, k); X^{*}X \leq I_k\} .$$
\begin{rem}
	We selected a frame $F^{\prime }$ among the coset $FU(k)$ under the condition that $Y^{\prime }$ is nonnegative.
	But this condition can not determine a unique frame $F^{\prime }$, because there are many elements $U \in U(k)$ such that the restriction $U_{\vert N(V)^{\bot }}$ of $U$ to the orthogonal complement $N(V)^{\bot }$ of $N(V)$ is $V$.
\end{rem}
\begin{defn}
	Define the mapping $\bar{\kappa }_{e}: \bar{B}(n-k, k) \rightarrow  G(k, \mathbb C^{n})$ by\newline
	$$    \bar{\kappa }_{e}: \bar{B}(n-k, k) \ni  X \rightarrow  \pi _{2} \begin{pmatrix} X \cr (I_{k} - X^{*}X)^{1/2} \end{pmatrix} \in G(k, \mathbb C^{n}) .$$ 
\end{defn}
\begin{rem}  \label{bd:kappa}
	The mapping
	$$    \kappa _{e}: B(n-k, k) \ni  X \rightarrow  \pi _{2} \begin{pmatrix} X \cr (I_{k} - X^{*}X)^{1/2} \end{pmatrix} \in  \Omega .$$ 
	of Proposition \ref{gr:kappa} is bijective.  But the mapping $\bar{\kappa }_{e}$ of the above definition
	is not bijective but only surjective.
\end{rem}
For $X \in  \bar{B} (n-k, k)$ denote
$$    W(X) = \begin{pmatrix} (I - XX^{*})^{1/2} & X \cr -X^{*}& Y
\end{pmatrix}, \  Y = (I - X^{*}X)^{1/2}. $$ 

\begin{prop}  \label{bd:section}
	For any $g \in  U(n)$ we can find $X \in  \bar{B} (n-k,
	k)$, $V \in  U(k)$ and  $h \in U(n-k) \times  I_{k}$
	such that\newline
	$$    g = \begin{pmatrix} (I - XX^{*})^{1/2} & X \cr -X^{*}& Y
	\end{pmatrix} 
	\begin{pmatrix} I_{n-k} & O \cr O & V \end{pmatrix} h$$ 
	$$    = \begin{pmatrix} (I - XX^{*})^{1/2} & X \cr -X^{*}& Y
	\end{pmatrix} 
	\begin{pmatrix} U & O \cr O & V \end{pmatrix},$$ 
	with $U \in  U(n-k)$ and $\bar{B} (m, k) = \{ X\in M(m, k); X^{*}X \leq  I_{k}\} $.\newline
\end{prop}
\begin{rem}  \label{bd:unique}
	The above proposition is the counterpart of Proposition \ref{gr:uniqueexpression}.  Note that $V \in  U(k)$ is not unique, and consequently, $X$ and $U$ are also not unique.  If we use $B(n-k, k)$ instead of $\bar{B}(n-k, k)$, then we have the uniqueness.  The above proposition  only holds if $g$ satisfies $\pi (g) \in \Omega $. The question of uniqueness is addressed in the following proposition.
\end{rem}

\begin{prop}  \label{bd:unique}
	Let $g, g^{\prime } \in  U(n-k_{1}- \cdots -k_{j})$ and $W(X_{j}),
	W(X^{\prime }_{j})$ for $X_{j}, X^{\prime }_{j}  \in  \bar{B} (n-k_{1}- \cdots -k_{j}, k_{j})$ such that $W(X_{j})g =
	W(X^{\prime }_{j})g^{\prime }$.  Then $X_{j} = X^{\prime }_{j}$ and $g = g^{\prime }$.
\end{prop}
\begin{proof}  Let
	$$    \pi _1: U(n-k_{1}- \cdots -k_{j}) \rightarrow  S(k_{j}, \mathbb C^{n-k_{1}- \cdots -k_{j}})$$  
	be the projection defined by (\ref{gr:pi1}).  Then $\begin{pmatrix} X_{j} \cr Y_{j}
	\end{pmatrix}$ = $\pi (W(X_{j})g)$ = $\pi (W(X^{\prime }_{j})g^{\prime })$ = $\begin{pmatrix} X^{\prime }_{j}
	\cr Y^{\prime }_{j} \end{pmatrix}$, and $X_{j} = X^{\prime }_{j}$, $g = W(X_{j})^{-1}W(X^{\prime }_{j})g^{\prime } = g^{\prime }$.
	
	It follows from Proposition \ref{bd:section} that for any $U_{1} \in  U(n-k_{1})$, there is
	$W(X_{2})$ for $X_{2} \in  \bar{B} (n-k_{1}-k_{2}, k_{2})$ such that
	$$    U_{1} = W(X_{2}) \begin{pmatrix} U_{2} & O \cr O & V_{2} \end{pmatrix}
	= \begin{pmatrix} (I - X_{2}X_{2}^{*})^{1/2}U_{2} & X_{2}V_{2} \cr -X_{2}^{*}U_{2} & Y_{2}V_{2}
	\end{pmatrix},$$ 
	where $U_{2} \in  U(n-k_{1}-k_{2})$ and $V_{2} \in  U(k_{2})$.  Then  (5) and (6) imply, for $g \in  U(n)$,
	$$    g = W(X_{1})
	\begin{pmatrix}
	\begin{pmatrix} (I - X_{2}X_{2}^{*})^{1/2}U_{2} & X_{2}V_{2} \cr -X_{2}^{*}U_{2} & Y_{2}V_{2} \end{pmatrix} 
	& O \cr O & V_{1} \end{pmatrix}$$ 
	$$    = W(X_{1})
	\begin{pmatrix}
	\begin{pmatrix} (I - X_{2}X_{2}^{*})^{1/2} & X_{2} \cr -X_{2}^{*} & Y_{2} \end{pmatrix} 
	& O \cr O & I_{k_{1}} \end{pmatrix}
	\begin{pmatrix}
	U_{2} & O  & O \cr O & V_{2} & O \cr
	O  & O &  V_{1} \end{pmatrix}.$$  
\end{proof}
Iteration of this procedure gives
\begin{prop}
	(1)  For any $g \in  U(n)$ there exists $(X_{m}, \ldots , X_{2}) \in  \bar{B} (n-k_{m},
	k_{m}) \times  \cdots \times  \bar{B} (n-k_{3}- \cdots - k_{m}, k_{2})$ and $(V_{1}, \ldots , V_{m}) \in 
	U(k_{1}) \times \cdots \times  U(k_{m})$ such that\newline
	$$    g = W(X_{m}) 
	\begin{pmatrix} W(X_{m-1}) & O \cr O & I_{k_{m}} \end{pmatrix} \cdots
	\begin{pmatrix} W(X_{2}) & O \cr O & I_{k_{3}+ \cdots + k_{m}} \end{pmatrix} 
	\begin{pmatrix} V_{1} & \ldots & O \cr \vdots & \ddots & \vdots \cr O & \ldots & V_{m} \end{pmatrix} .$$ 
	(2)  The mapping\newline
	$$    \bar{B} (n-k_{m}, k_{m}) \times  \cdots \times  \bar{B} (n-k_{3}- \cdots - k_{m}, k_{2}) \ni  (X_{m}, \ldots , X_{2}) \rightarrow  $$ 
	$$    W(X_{m}) 
	\begin{pmatrix} W(X_{m-1}) & O \cr O & I_{k_{m}} \end{pmatrix} \cdots
	\begin{pmatrix} W(X_{2}) & O \cr O & I_{k_{3}+ \cdots + k_{m}} \end{pmatrix} \in  U(n) $$ 
	is injective.
\end{prop}
\begin{proof}
	The existence of $X_{j}$ and $V_{j}$ follows from Proposition \ref{bd:section} and the
	injectivity of the mapping follows from Proposition \ref{bd:unique}.
\end{proof}
\begin{rem}
	The above proposition is the counterpart of Propositon \ref{gr:canonicaldec} and the following proposition is the counterpart of Proposition \ref{pr:canonicaldec}.
\end{rem}
\begin{prop}
	Let $x_j \in \bar{B}(j-1) = \{x \in \mathbb C ^{j-1}; \norm{x} \leq 1\}$ and
	\beq  \label{bd:Wxj}
	W(x_{j}) = \begin{pmatrix} (I_{j-1} - x_{j}x_{j}^{*})^{1/2} & x_{j} \cr -x_{j}^{*}& y_{j}
	\end{pmatrix}, \  y_{j} = (1 - x_{j}^{*}x_{j})^{1/2} .
	\eeq
	(1)  For any $g \in  U(n)$ there exists $(x_{n}, x_{n-1}, \ldots , x_{2}) \in  \bar{B} (n-1) \times \bar{B} (n-2) \times  \cdots \times  \bar{B} (1)$ and $(V_{1}, \ldots , V_{n}) \in 
	U(1) \times \cdots \times  U(1)$ such that
	\beq  \label{bd:parameter}
	g = W(x_{n}) 
	\begin{pmatrix} W(x_{n-1}) & O \cr O & I_{1} \end{pmatrix} \cdots
	\begin{pmatrix} W(x_{2}) & O \cr O & I_{n-2} \end{pmatrix} 
	\begin{pmatrix} V_{1} & \ldots & O \cr \vdots & \ddots & \vdots \cr O & \ldots & V_{n} \end{pmatrix} .
	\eeq 
	(2)  The mapping\newline
	$$    \bar{B} (n-1) \times  \cdots \times  \bar{B} (1) \ni  (x_{n}, \ldots , x_{2}) \rightarrow  $$ 
	$$    W(x_{n}) 
	\begin{pmatrix} W(x_{n-1}) & O \cr O & I_{1} \end{pmatrix} \cdots
	\begin{pmatrix} W(x_{2}) & O \cr O & I_{n-2} \end{pmatrix} \in  U(n) $$ 
	is injective.
\end{prop}

Instead of $\bar{B}(j-1)$ in \cite{BP08,BMMP12}
the following parameter space $\bar{Q}_{j}$ 
\beq  \label{bd:Qj}
\bar{Q}_{j} = \{ (\theta _{j}, \zeta _{j}); 0 \leq  \theta  \leq  \pi /2, \zeta _{j} \in  S(\mathbb C^{j-1})\}, \ S(\mathbb C^{j-1}) = \{\zeta \in \mathbb C^{j-1}; \norm{\zeta } = 1\}
\eeq
is introduced for the Jarlskog parametrization \cite{Ja05}, \cite{Ja06}.
The mapping
$$      \bar{Q}_{j} \ni  (\theta _{j}, \zeta _{j}) \rightarrow  \sin \theta _{j} \zeta _{j} = x \in  \bar{B}(j-1) = \{ x \in  \mathbb C^{j-1}; \Vert x\Vert  \leq  1\} $$
shows that both parameter spaces are the same.

The formula in \cite{Ja05} which correspond to the formula (\ref{bd:parameter}) is
$$  U_n = A_{n,n}A_{n,n-1} \cdots A_{n,2}D(e^{i \alpha _1}, \ldots e^{i \alpha _n}),$$
where, using bra and ket notation of Physics,
$$  A_{n,j} = \begin{pmatrix} V_{j}(\theta _{j}, \zeta _{j}) & 0 \cr 0 & I_{n-j} \end{pmatrix},
V_{j}(\theta _{j}, \zeta _{j}) = \begin{pmatrix} I_{j-1} - (1 - \cos \theta _{j})\vert \zeta _{j}\rangle \langle \zeta _{j}\vert  & \sin \theta _{j} \vert \zeta _{j}\rangle  \cr
- \sin \theta _{j} \langle \zeta _{j}\vert  & \cos \theta _{j} \end{pmatrix} . $$ 
For  $x=\sin{\theta_j} \zeta_j$, $V_{j}(\theta _{j}, \zeta _{j})$ and $W(x)$ of (\ref{bd:Wxj}) are precisely the same.\newline

\indent In \cite{Ak07b} there is a statement that a typical coset
representative in the coset space $U(2)/(U(1)\times U(1))$ is\newline
$$    V = \begin{pmatrix} \cos \alpha  & e^{i\phi } \sin \alpha  \cr -e^{-i\phi } \sin \alpha  & \cos
\alpha  \end{pmatrix}. $$
$V$ is $V_{1}(\theta _{1}, \zeta _{1})$ for $(\theta _1, \zeta _1) = (\alpha , \phi ) \in \bar{Q}_1$.
But if $\cos \alpha  = 0$ and $e^{i\phi } \neq 1$,\newline
$$    \begin{pmatrix} 0 & e^{i\phi } \cr -e^{-i\phi } & 0  \end{pmatrix}
\neq \begin{pmatrix} 0 & 1 \cr - 1 & 0  \end{pmatrix} $$ 
and\newline
$$    \begin{pmatrix} 0 & e^{i\phi } \cr -e^{-i\phi } & 0  \end{pmatrix}
= \begin{pmatrix} 0 & e^{i\phi } \cr -e^{-i\phi } & 0  \end{pmatrix}
\begin{pmatrix} 1 & 0 \cr 0 & 1  \end{pmatrix}
= \begin{pmatrix} 0 & 1 \cr - 1 & 0  \end{pmatrix} 
\begin{pmatrix} e^{-i\phi } & 0 \cr 0 & e^{i\phi }  \end{pmatrix}.$$
This is an example showing that the expression of Proposition \ref{bd:section} is not unique. 
We find
$$    \begin{pmatrix} 0 & e^{i\phi } \cr -e^{-i\phi } & 0  \end{pmatrix}
\begin{pmatrix} \lambda  & 0 \cr 0 & \mu   \end{pmatrix} 
\begin{pmatrix} 0 & e^{i\phi } \cr -e^{-i\phi } & 0  \end{pmatrix} ^{*}= \begin{pmatrix} \mu &0\\0& \lambda \end{pmatrix}
= \begin{pmatrix} 0 & 1 \cr -1 & 0  \end{pmatrix}
\begin{pmatrix} \lambda  & 0 \cr 0 & \mu   \end{pmatrix} 
\begin{pmatrix} 0 & 1 \cr -1 & 0  \end{pmatrix} ^{*},$$ 
which shows that the density matrix $\begin{pmatrix} \mu  & 0 \cr 0 & \lambda   \end{pmatrix} $ has two parametrization, i.e., $((\lambda , \mu ), 1)$ and $((\lambda , \mu ), e^{i \alpha })$.

However, if we use the parameter space 
$$  Q_{j} = \{ (\theta _{j}, \zeta _{j}); 0 \leq  \theta  <  \pi /2, \zeta _{j} \in  S(\mathbb C^{j-1})\}, \ S(\mathbb C^{j-1}) = \{\zeta \in \mathbb C^{j-1}; \norm{\zeta } = 1\} $$
instead of  $\bar{Q}_{j}$ of (\ref{bd:Qj}), then the uniqueness is recovered.  But not every $g$ can not be expressed by such parameters
(see Remark \ref{bd:unique}).
Let $g = \begin{pmatrix} 0 & e^{i\phi } \cr -e^{-i\phi } & 0  \end{pmatrix}$.  Then $\pi (g) = \pi _2 \begin{pmatrix} e^{i\phi } \cr 0 \end{pmatrix}$ does not belong to $\Omega  = \Omega _2$ but belongs to $\Omega _1 = \pi _2 (\tilde {\Omega }_1)$, where $\tilde {\Omega }_j = \{z = (z_1, z_2) \in S(\mathbb C^2); z_j \neq 0\}$.  The unitary matrix $U_1$ in Remark \ref{pr:usigma} is 
$\begin{pmatrix} 0 & 1 \cr 1 & 0  \end{pmatrix}$, and the mapping $\kappa _1$ which gives bijection between $B(1)$ and $\Omega _1$ is $\kappa _1 = U_1 \kappa $:
$$  \kappa _1: B(1) \ni 0 \rightarrow \begin{pmatrix} 0 & 1 \cr 1 & 0  \end{pmatrix} 
\pi _2 \begin{pmatrix} 0 \cr 1 \end{pmatrix} = \pi _2\begin{pmatrix} 1 \cr 0 \end{pmatrix}
= \pi _2 \begin{pmatrix} e^{i\phi } \cr 0 \end{pmatrix},$$
and
$$  \lambda _1 = U_1 \lambda : B(1) \ni 0 \rightarrow   \begin{pmatrix} 0 & 1 \cr 1 & 0  \end{pmatrix}
\begin{pmatrix} 1 & 0 \cr 0 & 1  \end{pmatrix} = \begin{pmatrix} 0 & 1 \cr 1 & 0  \end{pmatrix}
= \iota _1 (\pi (g)).
$$
Thus $g$ has the unique form
$$  g = \begin{pmatrix} 0 & e^{i\phi } \cr -e^{-i\phi } & 0  \end{pmatrix} = \iota _1 (\pi (g)) h
= \begin{pmatrix} 0 & 1 \cr 1 & 0  \end{pmatrix} \begin{pmatrix} -e^{-i\phi } & 0 \cr 0 & e^{-i\phi }  \end{pmatrix}.  $$

It follows from Remarks \ref{bd:kappa} and \ref{bd:unique} that Proposition \ref{bd:section} for the parameter space	
$B(n-k, k)$	is valid only if $\pi (g) \in \Omega $.  As shown in the following, $G(k, \mathbb C^{n}) \backslash  \Omega  = \partial \Omega $.  So, in most cases $\partial \Omega _{\sigma }$ is negligible.

\begin{defn}
	Let $(\Omega _{\sigma }, \phi _{\sigma })_{\sigma \in S}$ be the system of (at most) countable coordinate neighborhoods of a $m$-dimensional manifold $M$.  A subset $A \subset M$ is said to have the measure zero if for every coordinate neighborhood $(\Omega _{\sigma }, \phi _{\sigma })$ the set $\phi  _{\sigma }(A \cap \Omega _{\sigma })$ has Lebesque measure zero in $\mathbb R^{m}$.
\end{defn}

\begin{prop}
	Let $(\Omega _{\sigma }, \phi _{\sigma })_{\sigma \in S}$ be the atlas of the Grassmann manifold $G(k, \mathbb C^{n})$.  Then
	$\Omega _{\sigma } \cap  \Omega _{\sigma ^{\prime }}$ is an open and dense subset of $\Omega _{\sigma }$.\newline
\end{prop}

\begin{proof}  Assume $\sigma  \neq \sigma ^{\prime }$.  Let\newline
	$$    M_{\sigma } M_{\sigma (n-k+1), \ldots , \sigma (n)}^{-1} = 
	\begin{pmatrix} Z \cr I \end{pmatrix}      
	= \begin{pmatrix} z_{11} & \ldots & z_{1k} \cr
	\vdots & \ddots & \vdots \cr
	z_{n-k,1} & \ldots & z_{n-k,k} \cr
	1   & \ldots & 0 \cr
	\vdots & \ddots & \vdots \cr
	0   & \ldots & 1  \end{pmatrix} = N.$$ 
	Define a function $f(Z)$ of $Z \in  M(n-k, k)$ by 
	$$   f(Z) = \det N_{\sigma ^{-1} \circ \sigma ^{\prime }(n-k+1), \sigma ^{-1} \circ \sigma ^{\prime }(n-k+2), \ldots , \sigma ^{-1} \circ \sigma ^{\prime }(n)}.$$
	Then $f(Z)$ is a non-constant polynomial (provided $\sigma  \neq \sigma ^{\prime }$), and
	$$    \Omega _{\sigma } \supset  \Omega _{\sigma } \cap  \Omega _{\sigma ^{\prime }} = \{ Z \in  M(n-k, k); f(Z) \neq 0\} , $$ 
	$$    \Omega _{\sigma } \backslash  \Omega _{\sigma ^{\prime }} = \{ Z \in  M(n-k, k); f(Z) = 0\} .$$ 
	Since $f(Z)$ is a polynomial of $Z$, $f(Z)$ is continuous and $\Omega _{\sigma } \cap  
	\Omega _{\sigma ^{\prime }}$ is open subset of $\Omega _{\sigma }$.  For any neighborhood $V$ of $Z_{0} \in  \Omega _{\sigma } \backslash 
	\Omega _{\sigma ^{\prime }}$ there exists $Z_{1} \in  V$ such that $Z_{1} \in  \Omega _{\sigma } \cap  \Omega _{\sigma ^{\prime }}$ ($\Omega _{\sigma } \cap  \Omega _{\sigma ^{\prime }}$ is
	dense in $\Omega _{\sigma }$).  Otherwise, there exists a neighborhood $V$ of $Z_{0}$
	such that $f(Z) = 0$ on $V$.  Since $f(Z)$ is an analytic function,
	we have $f(Z) \equiv  0$.  This is a contradiction.
\end{proof}
\begin{cor}
	$\Omega _{\sigma }$ is an open and dense subset of $G(k, \mathbb C^{n})$, and
	therefore $G(k, \mathbb C^{n}) \backslash  \Omega _{\sigma } = \partial \Omega _{\sigma }$ and $\partial \Omega _{\sigma }$ is a  set of measure zero.\newline
\end{cor}
\begin{proof}  This follows from Relation  (\ref{gr:union}).\end{proof}

{
	Note that $(U(n), G(k, \mathbb C^{n}), \pi , U(n-k) \times  U(k))$ has
	the structure of a fiber bundle, where $U(n)$, $G(k, \mathbb C^{n})$, and
	$U(n-k) \times  U(k)$ are  the total space, the base space, and the fiber
	respectively, and $\pi : U(n) \rightarrow  G(k, \mathbb C^{n})$ is a continuous surjection
	satisfying a local triviality condition:\newline
	For every $z \in  G(k, \mathbb C^{n})$, there is an open neighborhood $\Omega _{\sigma }$ of $z$
	(which will be called a trivializing neighborhood) such that there
	is a homeomorphism
	$$   \phi : \Omega _{\sigma } \times  (U(n-k) \times  U(k)) \ni  (z, h) \rightarrow   \phi (z, h) = \iota _{\sigma }(z) h \in  \pi ^{-1}(\Omega _{\sigma }). $$ 
	\begin{prop}  $\pi ^{-1}(\partial \Omega _{e})$ is a set of measure zero.
	\end{prop}
	\begin{proof}
		Since $\partial \Omega _{e} \cap  \Omega _{\sigma }$ is the boundary of $\Omega _{e} \cap  \Omega _{\sigma }$ in $\Omega _{\sigma }$, $\partial \Omega _{e} \cap  \Omega _{\sigma } \times 
		(U(n-k) \times  U(k)) \cong  \pi ^{-1}(\partial \Omega _{e} \cap  \Omega _{\sigma })$ is the boundary of $\Omega _{e} \cap  \Omega _{\sigma } \times 
		(U(n-k) \times  U(k)) \cong  \pi ^{-1}(\Omega _{e} \cap  \Omega _{\sigma })$ in $\Omega _{\sigma } \times  (U(n-k) \times  U(k)) \cong  
		\pi ^{-1}(\Omega _{\sigma })$.  Thus $\pi ^{-1}(\partial \Omega _{e})$ is the boundary of $\pi ^{-1}(\Omega _{e})$ in $U(n) =
		\pi ^{-1}(G(k, \mathbb C^{n}))$ and a set of measure zero.
	\end{proof}
}
\begin{prop}  \label{bd:section}
	For almost all $g \in U(n)$, there is a unique $X \in  B (n-k,
	k)$ and $h \in U(n-k) \times U(k)$ such that 
	$$
	g = W(X) h.
	$$
\end{prop}
{
	\begin{proof}
		The proposition follows from the fact that $\pi ^{-1}(\partial \Omega _{e}) \cup \pi ^{-1}(\Omega _{e}) = U(n)$ and the previous proposition.
	\end{proof}
	\begin{prop}
		For almost all $g \in  U(n)$ we have a mapping $U(n) \ni g \rightarrow (X_{3}, X_{2}) \in  B(n-k_{3},
		k_{3}) \times B(n-k_{3}-k_{2}, k_{2})$ and a unique $h \in 
		U(k_{1}) \times  U(k_{1}) \times  U(k_{3})$ such that
		\beq  \label{bd:twotimes}
		g = W(X_{3}) 
		\begin{pmatrix} W(X_{2}) & O \cr O & I_{k_{3}} \end{pmatrix} h.
		\eeq 
	\end{prop}
	\begin{proof}
		Let $m = n-k_{3}$ and consider the two fiber bundles $F = (U(n), G(k_{3}, \mathbb C^{n}),
		\pi , U(n-k_{3}) \times  U(k_{3}))$ and $F^{\prime } = (U(m), G(k_{2}, \mathbb C^{m}),
		\pi ^{\prime }, U(m-k_{2}) \times  U(k_{2}))$. 
		Let $\Omega _{e}^{\prime }$ be a trivializing neighborhood of
		$F^{\prime }$.  Then $\partial \Omega _{e}^{\prime } \times  (U(m-k_{2}) \times  U(k_{2})) \cong  \pi ^{\prime -1}(\partial \Omega _{e}^{\prime })$ is the boundary
		of $\Omega _{e}^{\prime } \times  (U(m-k_{2}) \times  U(k_{2})) \cong  \pi ^{\prime -1}(\Omega _{e}^{\prime })$ in $U(m) $ and $\Omega _{e} \times 
		\pi ^{\prime -1}(\partial \Omega _{e}^{\prime }) \times  U(k_{3}) \cong  \phi (\Omega _{e}, \pi ^{\prime -1}(\partial \Omega _{e}^{\prime }) \times  U(k_{3}))$ is the boundary of
		$\Omega _{e} \times  \pi ^{\prime -1}(\Omega _{e}^{\prime }) \times  U(k_{3}) \cong  \phi (\Omega _{e}, \pi ^{\prime -1}(\Omega _{e}^{\prime }) \times  U(k_{3}))$ in $\Omega _{e} \times  U(m) \times 
		U(k_{3}) \cong  \pi ^{-1}(\Omega _{e})$.  Thus $\phi (\Omega _{e}, \pi ^{\prime -1}(\partial \Omega _{e}^{\prime }) \times  U(k_{3}))$ is a set of
		measure zero.  Since 
		$$    \phi (\Omega _{e}, \pi ^{\prime -1}(\Omega _{e}^{\prime }) \times  U(k_{3})) \cup \phi (\Omega _{e}, \pi ^{\prime -1}(\partial \Omega _{e}^{\prime }) \times  U(k_{3})) = \pi ^{-1}(\Omega _{e})  ,$$
		almost all $g \in  U(n)$ are expressed as
		$y = \phi (\Omega _{e}, \pi ^{\prime -1}(\Omega _{e}^{\prime }) \times  U(k_{3}))$, which is just (\ref{bd:twotimes}).
	\end{proof}
	Continuation of the above chain of arguments implies the following theorem.
}
\begin{thm}
	For almost all $g \in  U(n)$ we have a mapping $U(n) \ni g \rightarrow (X_{m}, \ldots , X_{2}) \in  B(n-k_{m},
	k_{m}) \times  \ldots \times B(n-k_{3}- \ldots - k_{m}, k_{2})$ and a unique $h \in 
	U(k_{1}) \times \cdots \times  U(k_{m})$ such that
	$$    g = W(X_{m}) 
	\begin{pmatrix} W(X_{m-1}) & O \cr O & I_{k_{m}} \end{pmatrix} \cdots
	\begin{pmatrix} W(X_{2}) & O \cr O & I_{k_{3}+ \cdots + k_{m}} \end{pmatrix} h.
	$$ 
\end{thm}

\section{Section on Grassmannian} \label{section}

The mapping
$$ 
W(X) = \begin{pmatrix} (I_{n-k} - XX^{*})^{1/2} & X \cr -X^{*} & (I_{k} -
X^{*}X)^{1/2} \end{pmatrix} 
$$
of (\ref{gr:WX}) for $X \in B(n-k, k)$ which gives a local section $\iota _{e} = W \circ \psi _{e} : \Omega _{e} \rightarrow U(n)$ is not really suitable for concrete  calculations.  So  here we construct a local section using simpler ones 
$$   
W(x) = \begin{pmatrix} (I_{n-1} - xx^{*})^{1/2} & x \cr -x^{*} & (1 -
x^{*}x)^{1/2} \end{pmatrix} . 
$$
of (\ref{pr:Wofx}) for $x \in B(n-1)$ which gives a local section 
$\iota _{n} = W \circ \psi _{n}: \Omega _{n} \rightarrow   U(n)$. 

We begin with the embedding of $CP^{m-1}$ in $CP^{n-1}$ ($m < n$).
$S(\mathbb C^{m})$ can be embeded in $S(\mathbb C^{n})$ by
$$      \iota _{n,m}: S(\mathbb C^{m}) \ni  (z_{1}, \ldots , z_{m})^{T} \rightarrow  (z_{1}, \ldots , z_{m}, 0, \ldots , 0)^{T} \in  S(\mathbb C^{n}).$$
This mapping $\iota _{n,m}$ induces the embedding $\iota _{n,m}: CP^{m-1} \rightarrow  CP^{n-1}$.\newline
Let $\sigma _{j,n} \in  S_{1,n}$ such that $\sigma _{j,n}(n) = j$.  Then $\tilde{\Omega } _{m} = \tilde{\Omega } _{m,m} =  \tilde{\Omega } _{\sigma _{m,m}}$ (resp. $\Omega _{m} = \Omega _{m,m} = \Omega _{\sigma _{m,m}}$) is
identified with $\tilde{\Omega } _{m,n} = \tilde{\Omega } _{\sigma _{m,n}}$ (resp. $\Omega _{m,n} = \Omega _{\sigma _{m,n}}$).
Let $\tilde{\psi } _{m,n}$ be the mapping
$$   \tilde{\psi } _{m,n}: \tilde{\Omega } _{m,n} \ni  z = (z_{1}, \ldots , z_{m-1}, 0, \ldots , 0, z_{m})^{T} \rightarrow  (z_{1},
\ldots , z_{m-1}, 0, \ldots , 0)^{T}/e^{i\theta } \in  B(n-1), \  z_{m} = \vert z_{m}\vert e^{i\theta }.$$    
Then $\tilde{\psi } _{m,n}$ induces the mapping
$$   \psi _{m,n}: \Omega _{m,n} \ni  \pi _{2}(z)  \rightarrow  (z_{1}, \ldots , z_{m-1}, 0, \ldots , 0)^{T}/e^{i\theta } \in 
B(n-1), \  z_{m} = \vert z_{m}\vert e^{i\theta },$$
where $\pi _{2}: S(\mathbb C^{m}) \rightarrow CP^{m-1}$ is the canonical projection. 
\begin{prop} \label{pr2:proptriangular}
	Let
	$$    S_{k}(\mathbb C^{n}) \ni  F = \begin{pmatrix} X \cr Y \end{pmatrix}, \  X \in 
	M(n-k, k), \  Y \in  M(k, k) $$ 
	be given with  $\det Y \neq 0$.  Then there exists a unique $U \in  U(k)$ such that
	$Y^{\prime } = YU = T$ where $T$ is a lower triangular matrix:
	\beq  \label{pr2:triangular}
	T = \begin{pmatrix} t_{11} & 0 & \ldots & 0 \cr
		t_{21} & t_{22} & \ddots & \vdots \cr
		\vdots & \ddots & \ddots & 0 \cr
		* & \ldots & * & t_{kk} \end{pmatrix}, \  t_{jj} > 0, \  1 \leq  j \leq  k. 
	\eeq 
\end{prop}
\begin{proof}  Let $\mathbb C^{*k}$ be the set of all complex row $k$ vectors $\boldsymbol{z} =
	(z_{1}, \ldots , z_{k})$ with an inner product $(\boldsymbol{z}, \boldsymbol{z}^{\prime }) = \sum _{j=1}^{k}
	z_{j}\bar{z} ^{\prime }_{j}$, and $\boldsymbol{z}^{*} = (\bar{z} _{1}, \ldots , \bar{z} _{k})^{T}$.  Let $\boldsymbol{y}_{j}$ be the $j$-th row
	vector of $Y$ and $\{ \boldsymbol{u}_{1}, \ldots , \boldsymbol{u}_{k}\} $ be the Schmidt's
	orthogonalization of $\{ \boldsymbol{y}_{1}, \ldots , \boldsymbol{y}_{k}\} $, i.e., $\boldsymbol{u}_{1} = \boldsymbol{y}_{1}/\Vert \boldsymbol{y}_{1}\Vert $, $\boldsymbol{u}_{j}
	= \boldsymbol{x}_{j}/\Vert \boldsymbol{x}_{j}\Vert $, $\boldsymbol{x}_{j} = \boldsymbol{y}_{j} - \sum _{i = 1}^{j-1} (\boldsymbol{y}_{j}, \boldsymbol{u}_{i})\boldsymbol{u}_{i}$.  Define $U \in  U(k)$ by \newline
	$$    U = (\boldsymbol{u}_{1}^{*}, \ldots , \boldsymbol{u}_{k}^{*}).  $$ 
	Since ${\rm span \, }\{ \boldsymbol{y}_{1}, \ldots , \boldsymbol{y}_{j}\}  = {\rm span \, }\{ \boldsymbol{u}_{1}, \ldots , \boldsymbol{u}_{j}\} $, $(\boldsymbol{y}_{i}, \boldsymbol{u}_{l}) =
	0$ if $i < l$, and $(\boldsymbol{y}_{j}, \boldsymbol{u}_{j}) > 0$.  Thus we have $YU = T$.  
	For the uniqueness of $U$,
	suppose $Y = TU^* = T^{\prime }U^{\prime *}$ for
	$$   U^* = \begin{pmatrix} \boldsymbol{u}_{1} \cr \boldsymbol{u}_{2} \cr \vdots \cr \boldsymbol{u}_{k} \end{pmatrix},
	U^{\prime *} = \begin{pmatrix} \boldsymbol{u}_{1}^{\prime } \cr \boldsymbol{u}_{2}^{\prime } \cr \vdots \cr \boldsymbol{u}_{k}^{\prime } \end{pmatrix},
	T^{\prime } = \begin{pmatrix} t_{11}^{\prime } & 0 & \ldots & 0 \cr
	t_{21}^{\prime } & t_{22}^{\prime } & \ddots & \vdots \cr
	\vdots & \ddots & \ddots & 0 \cr
	* & \ldots & * & t_{kk}^{\prime } \end{pmatrix}. $$ 
	Then we have $\boldsymbol{y}_{1} = t_{11} \boldsymbol{u}_{1} = t_{11}^{\prime } \boldsymbol{u}_{1}^{\prime }$ and therefore $t_{11} = t_{11}^{\prime }$ and $\boldsymbol{u}_{1} = \boldsymbol{u}_{1}^{\prime }$.  From $\boldsymbol{y}_{2} = t_{21} \boldsymbol{u}_{1} + t_{22} \boldsymbol{u}_{2} = t_{21}^{\prime } \boldsymbol{u}_{1} + t_{22}^{\prime } \boldsymbol{u}_{2}^{\prime }$, we have $t_{21} = (\boldsymbol{y}_{2}, \boldsymbol{u}_{1}) = t_{21}^{\prime }$ and $\boldsymbol{y}_{2} - t_{21} \boldsymbol{u}_{1} = t_{22} \boldsymbol{u}_{2} = t_{22}^{\prime } \boldsymbol{u}_{2}^{\prime }$.  Thus we have $t_{22} = t_{22}^{\prime }$ and $\boldsymbol{u}_{2} = \boldsymbol{u}_{2}^{\prime }$.  Continuing these procedures, we get $U^* = U^{\prime *}$, i.e., the uniqueness of $U$.
\end{proof}
\begin{prop}  \label{pr2:recursive}
	Let $g_{n} \in  U(n)$ be given with
	\beq \label{pr2:W}
	g_{n} = \begin{pmatrix} W & X \cr V & T \end{pmatrix}, 
	\eeq
	where $W \in  M(n-k, n-k)$, $X \in  M(n-k, k)$, $V \in  M(k, n-k)$ and $T \in 
	M(k, k)$ is a lower triangular matrix with positive diagonal
	elements as (\ref{pr2:triangular}). 
	Then there exist $x \in \mathbb C^{n-1}$ and $g_{n-1} \in  U(n-1)$ such that
	\beq \label{pr2:eqrecursive} 
	g_{n} = W(x)\cdot (g_{n-1} \times  I_{1}),
	\eeq 
	where
	$$    W(x) = \begin{pmatrix} (I_{n-1} - xx^{*})^{1/2} & x \cr -x^{*} & (1 -
	x^{*}x)^{1/2} \end{pmatrix}, $$
	and $g_{n-1}$ has the form
	\beq \label{pr2:W2}
	g_{n-1} = \begin{pmatrix} W^{\prime } & X^{\prime } \cr V^{\prime } & T^{\prime } \end{pmatrix}, 
	\eeq
	where $W^{\prime } \in  M(n-k, n-k)$, $X^{\prime } \in  M(n-k, k-1)$, $V^{\prime } \in  M(k-1, n-k)$ and
	$T^{\prime } \in 
	M(k-1, k-1)$ is a lower triangular matrix with positive diagonal
	elements of the form
	\beq  \label{pr2:triangular2}
	T = \begin{pmatrix} T^{\prime } & 0 \cr * & t_{kk} \end{pmatrix}.
	\eeq
\end{prop}
\begin{proof}
	We only have  to show that $T^{\prime } \in  M(k-1, k-1)$ is of the form of (\ref{pr2:triangular2}).  Let $(x, t)^{T}$
	$(x \in  \mathbb C^{n-1}, t \in  \mathbb C)$ be the last column of the matrix $g_{n}$.
	
	Since $T$ is lower triangular, $x = (x_{1}, \ldots , x_{n-k}, 0, \ldots ,
	0)^{T}$, $(I_{n-1} - xx^{*})^{1/2} = (I_{n-k} - x^{\prime }x^{\prime *})^{1/2} \times  I_{k-1}$, where $x^{\prime } = (x_{1},
	\ldots , x_{n-k})^{T}$.  Since $W(x)^{-1}g_{n} = g_{n-1} \times  I_{1}$,
	$$    T^{\prime } \times  I_{1} = \begin{pmatrix} O_{1} & I_{k-1} & O_{2} \cr
	x^{*} & O_{3} & (1 - x^{*}x)^{1/2} \end{pmatrix}
	\begin{pmatrix} X \cr T \end{pmatrix} , $$ 
	where $O_{1}$ (resp. $O_{2}$, $O_{3}$) is the $(k-1) \times  (n-k)$ (resp. $(k-1) \times 
	1$, $1 \times  (k-1)$) zero matrix.  Therefore \newline
	$$     \begin{pmatrix} T^{\prime } & O_{4} \end{pmatrix}
	= \begin{pmatrix} I_{k-1} & O_{2} \end{pmatrix} T, $$ 
	where $O_{4}$ is the $(k-1) \times  1$ zero matrix and $T^{\prime }$ is the $(k-1) \times  (k-1)$-submatrix of $T$.
\end{proof}
\begin{prop}    
	Let $\pi : U(n) \rightarrow G(k, \mathbb C^{n})$ be the canonical projection.	
	Then there is a unique surjection 
	$$
	f: U(n) \supset \pi ^{-1}(\Omega _{e}) \ni g \rightarrow (z_{k}, z_{k-1}, \ldots , z_{1}) 
	$$
	$$
	\in  \Omega _{m, n} 
	\times \Omega _{m, n-1} \times \ldots \times \Omega _{m, n-k+1} \subset CP^{n-1} \times CP^{n-2} \times \ldots \times CP^{n-k}
	$$ 
	for $m = n-k+1$ and $\Omega _{e} \subset  G(k, \mathbb C^{n})$, and a unique $h \in U(n-k) \times U(k)$ such that
	\beq \label{pr2:cosetrep}
	g = W(\psi _{m,n}(z_{k})) \begin{pmatrix} W(\psi _{m,n-1}(z_{k-1})) & 0 \cr 0 & I_{1} \end{pmatrix} \cdots 
	\begin{pmatrix} W(\psi _{m,n-k+1}(z_{1})) & 0 \cr 0 & I_{k-1} \end{pmatrix} h.
	\eeq
\end{prop}
\begin{proof}
	Since $g \in  \pi ^{-1}(\Omega _{e})$, $g$ has the form of (\ref{pr2:W}) with $\det Y \neq 0$ and there exists $U \in  U(k)$ such that $T = YU$ has the form of (\ref{pr2:triangular}).  Let \newline
	$$    g_{n} = g \begin{pmatrix} I_{n-k} & 0 \cr 0 & U \end{pmatrix}.  $$ 
	Then it follows from Proposition \ref{pr2:recursive} that there exists $g_{n-1} \in  U(n-1)$
	which satisfies (\ref{pr2:eqrecursive}) and has the form of (\ref{pr2:W2}) where $T^{\prime }$ again satisfies
	the hypothesis of Proposition \ref{pr2:recursive}.  Iterating this argument, we get
	$$
	g_{n} = W(\psi _{m,n}(z_{k})) \begin{pmatrix} W(\psi _{m,n-1}(z_{k-1})) & 0 \cr 0 & I_{1} \end{pmatrix} \cdots 
	\begin{pmatrix} W(\psi _{m,n-k+1}(z_{1})) & 0 \cr 0 & I_{k-1} \end{pmatrix} h,
	$$
	where $h = g_{n-k} \times I_{k}, \ g_{n-k} \in U(n-k)$.  This shows (\ref{pr2:cosetrep}) with $h = g_{n-k} \times U$.
	The relation 
	$f(g^{\prime }) = (z_{k}, z_{k-1}, \ldots , z_{1})$ for $(z_{k}, z_{k-1}, \ldots , z_{1}) \in  \Omega _{m, n} 
	\times \Omega _{m, n-1} \times \ldots \times \Omega _{m, n-k+1}$ and
	$$
	g^{\prime } = W(\psi _{m,n}(z_{k})) \begin{pmatrix} W(\psi _{m,n-1}(z_{k-1})) & 0 \cr 0 & I_{1} \end{pmatrix} \cdots 
	\begin{pmatrix} W(\psi _{m,n-k+1}(z_{1})) & 0 \cr 0 & I_{k-1} \end{pmatrix} 
	$$
	show the surjectivity of $f$.	
\end{proof}
\begin{cor}  \label{pr2:canonicaldec}
	Let $\pi : U(n) \rightarrow G(k, \mathbb C^{n})$ be the canonical projection, and
	$\iota _{j}$ the section of $U(j)$
	on $CP^{j-1}$ defined by (\ref{pr:section}).
	Then there is a unique bijection 
	$$   \phi _{e}: G(k, \mathbb C^{n}) \supset \Omega _{e} \ni \pi (g) \rightarrow $$
	$$(z_{k}, z_{k-1}, \ldots , z_{1}) \in  \Omega _{m, n} 
	\times \Omega _{m, n-1} \times \ldots \times \Omega _{m, n-k+1} \subset CP^{n-1} \times CP^{n-2} \times \ldots \times CP^{n-k}  $$
	such that 
	$$ \pi (g) = \pi (g^{\prime }) {\rm \ for \ } g^{\prime } = \psi (\phi _{e}(\pi (g))),$$
	where
	$$  \psi (z_{k}, z_{k-1}, \ldots , z_{1}) = \iota _{n} (z_{k}) \begin{pmatrix} \iota _{n-1}(z_{k-1}) & 0 \\ 0 & I_{1} \end{pmatrix} 
	\cdots \begin{pmatrix} \iota _{n-k+1}(z_{1}) & 0 \\ 0 & I_{k-1}
	\end{pmatrix}
	. $$ 
	$\psi \circ \phi _{e}$ is a section of $U(n)$ on $\Omega _{e} \subset G(k, \mathbb C^{n})$ for $\pi $.
\end{cor}

\section{Lie algebraic back ground} \label{liealg}
In the articles which we have mentioned many statements are based
on the use of the Lie algebra $\mathfrak{u}(n)$ of Lie group $U(n)$. We comment here on the connection with the approach presented above.

The Lie algebra $\mathfrak{u}(n)$ of the Lie group $U(n)$ is defined by
$$      \mathfrak{u}(n) = \{ X \in  M(n, n); \forall t \in  \mathbb R, \exp tX \in  U(n)\} .$$
From the relation
$$      \exp tX^{*} = (\exp tX)^{*} = (\exp tX)^{-1} = \exp -tX $$
it follows
$$      \mathfrak{u}(n) = \{ X \in  M(n, n); X^{*} = -X\} .$$
Let $n = k_{1} + k_{2}$.  Then the Lie algebra of the Lie group $U(k_{1}) \times 
U(k_{2})$ is $\mathfrak{u}(k_{1}) \oplus  \mathfrak{u}(k_{2})$, namely, the set of
the elements of the form\newline
$$    \begin{pmatrix}X_{1} & 0 \cr 0 & X_{2} \end{pmatrix}, \  X_{j}
\in \mathfrak{u}(k_{j}).$$ 
Let $\mathfrak{p}$ be a subset of $\mathfrak{u}(n)$ such that\newline
$$    \mathfrak{u}(n) = \mathfrak{u}(k_{1}) \oplus  \mathfrak{u}(k_{2}) \oplus  
\mathfrak{p} . $$ 
Then $\mathfrak{p}$ consists of the elements of the form\newline
$$    K = K(B) = \begin{pmatrix}O_{1} & B \cr
-B^{*} &  O_{2}\end{pmatrix},$$ 
where $O_{j}$ ($j = 1, 2$) is the $k_{j} \times  k_{j}$ matrix whose
entries are all zero and $B$ is an $k_{1} \times  k_{2}$ complex matrix.  Since
the space
$$    \mathfrak{u}(n)/(\mathfrak{u}(k_{1}) \oplus  \mathfrak{u}(k_{2})) \cong
\mathfrak{p} = \{ K(B); B \in  M(k_{1}, k_{2})\}  $$ 
is considered to be the tangent space of the homogeneous space $U(n)/(U(k_{1}) \times  U(k_{2}))$ at $
o = \pi (e)$, where $e$ is the identity of $U(n)$ and $\pi : U(n) \rightarrow U(n)/(U(k_{1}) \times  U(k_{2}))$ of (\ref{gr:pi}), we study $\exp
K(B)$.  First, we have\newline
$$    K^{2} = \begin{pmatrix}-BB^{*} & O \cr
O^{*} & -B^{*}B\end{pmatrix},$$ 
where $BB^{*}$ is an $k_{1} \times  k_{1}$-matrix, $B^{*}B$ an $k_{2} \times  k_{2}$-matrix and
$O$ is the $k_{1} \times  k_{2}$-matrix whose entries are all zero. Observe now 
$$    K^{2n+2}  
= \begin{pmatrix}-BB^{*} & O \cr
O^{*} & -B^{*}B\end{pmatrix} ^{n+1} 
= \begin{pmatrix}-\sqrt{BB^{*}}^{2} & O \cr
O^{*} & -\sqrt{B^{*}B}^{2}\end{pmatrix} ^{n+1}  $$ 
$$    = \begin{pmatrix}(-1)^{n+1} \sqrt{BB^{*}}^{2n+2} & O \cr
O^{*} & (-1)^{n+1}\sqrt{B^{*}B}^{2n+2}\end{pmatrix}. $$ 
This gives
$$      I_{n} + \sum _{n=0}^{\infty } \frac{1}{(2n+2)!} K^{2n+2} 
= \begin{pmatrix}\cos \sqrt{BB^{*}} & O \cr
O^{*} & \cos \sqrt{B^{*}B}\end{pmatrix}$$ 
and similarly
$$    \sum _{n=0}^{\infty } \frac{1}{(2n+1)!} K^{2n}K
= \sum _{n=0}^{\infty } \frac{1}{(2n+1)!} \begin{pmatrix}(-BB^{*})^{n} & O \cr
O^{*} & (-B^{*}B)^{n}\end{pmatrix}
\begin{pmatrix}O_{1} & B \cr
-B^{*} &  O_{2}\end{pmatrix}$$ 
$$    = \sum _{n=0}^{\infty } \frac{1}{(2n+1)!} 
\begin{pmatrix}O_{1} & (-BB^{*})^{n}B \cr
(-B^{*}B)^{n}(-B^{*}) &  O_{2}\end{pmatrix}$$ 
$$    = \sum _{n=0}^{\infty } \frac{1}{(2n+1)!} 
\begin{pmatrix}O_{1} & B(-B^{*}B)^{n} \cr
(-B^{*}B)^{n}(-B^{*}) &  O_{2}\end{pmatrix}$$ 
$$    = \sum _{n=0}^{\infty } \frac{1}{(2n+1)!} 
\begin{pmatrix}O_{1} & B(-1)^{n} \sqrt{B^{*}B}^{2n} \cr
(-1)^{n}\sqrt{B^{*}B}^{2n}(-B^{*}) &  O_{2}\end{pmatrix}$$ 
$$    = \sum _{n=0}^{\infty } \frac{1}{(2n+1)!} 
\begin{pmatrix}O_{1} & B \sqrt{B^{*}B}^{-1}(-1)^{n} \sqrt{B^{*}B}^{2n+1} \cr
\sqrt{B^{*}B}^{-1}(-1)^{n}\sqrt{B^{*}B}^{2n+1}(-B^{*}) &  O_{2}\end{pmatrix}$$ 
$$    = \begin{pmatrix}O_{1} & B \sqrt{B^{*}B}^{-1} \sin \sqrt{B^{*}B} \cr
\sqrt{B^{*}B}^{-1} \sin \sqrt{B^{*}B} (-B^{*}) &  O_{2}\end{pmatrix}.$$ 
Thus we conclude\newline
$$      e^{K} = I_{n} + \sum _{n=0}^{\infty } \frac{1}{(2n+2)!} K^{2n+2} + \sum _{n=0}^{\infty } \frac{1}{(2n+1)!} K^{2n+1}=
\begin{pmatrix}\cos \sqrt{BB^{*}} & B \frac{\sin \sqrt{B^{*}B}}{\sqrt{B^{*}B}} \cr
- \frac{\sin \sqrt{B^{*}B}}{\sqrt{B^{*}B}} B^{*} &  \cos \sqrt{B^{*}B}\end{pmatrix}.$$ 

\begin{rem}  Since
	$$      (-BB^{*})^{n+1} = (-1)^{n+1} B(B^{*}B)^{n}B^{*},$$
	$$      \cos \sqrt{BB^{*}} = I_{k_{1}} + \sum _{n=0}^{\infty } \frac{1}{(2n+2)!} (-BB^{*})^{n+1} $$
	$$      = I_{k_{1}} + B \sum _{n=0}^{\infty } \frac{1}{(2n+2)!} (-1)^{n+1}(B^{*}B)^{n} B^{*}$$
	$$      = I_{k_{1}} + B (\sqrt{B^{*}B})^{-2} \sum _{n=0}^{\infty } \frac{1}{(2n+2)!} (-1)^{n+1}(\sqrt{B^{*}B})^{2n+2} B^{*}$$
	$$      = I_{k_{1}} + B (\sqrt{B^{*}B})^{-2} (\cos \sqrt{B^{*}B} - I_{k_{2}}) B^{*}.$$
\end{rem}
\begin{rem}
	Since
	$$      (-B^{*}B)^{n+1} = (-1)^{n+1} B^{*}(BB^{*})^{n}B,$$
	$$      \cos \sqrt{B^{*}B} = I_{k_{2}} + \sum _{n=0}^{\infty } \frac{1}{(2n+2)!} (-B^{*}B)^{n+1} $$
	$$      = I_{k_{2}} + B^{*} \sum _{n=0}^{\infty } \frac{1}{(2n+2)!} (-1)^{n+1}(BB^{*})^{n} B$$
	$$      = I_{k_{2}} + B^{*} (\sqrt{BB^{*}})^{-2} \sum _{n=0}^{\infty } \frac{1}{(2n+2)!} (-1)^{n+1}(\sqrt{BB^{*}})^{2n+2} B$$
	$$      = I_{k_{2}} + B^{*} (\sqrt{BB^{*}})^{-2} (\cos \sqrt{BB^{*}} - I_{k_{1}}) B.$$
\end{rem}
\begin{rem}
	$\cos \sqrt{B^{*}B}$, $\sqrt{B^{*}B}^{-1} \sin \sqrt{B^{*}B}$ and $(\sqrt{B^{*}B})^{-2} (\cos \sqrt{B^{*}B} -
	I_{k_{2}})$ are entire functions of $B^{*}B$.
\end{rem}
\begin{rem}
	$$    \sum _{n=0}^{\infty } \frac{1}{(2n+1)!} K^{2n}K
	= \sum _{n=0}^{\infty } \frac{1}{(2n+1)!} 
	\begin{pmatrix}O_{1} & (-BB^{*})^{n}B \cr
	(-B^{*}B)^{n}(-B^{*}) &  O_{2}\end{pmatrix}$$ 
	$$    = \sum _{n=0}^{\infty } \frac{1}{(2n+1)!} 
	\begin{pmatrix}O_{1} & (-BB^{*})^{n}B \cr
	(-B^{*})(-B^{*}B)^{n} &  O_{2}\end{pmatrix}$$ 
	$$    = \sum _{n=0}^{\infty } \frac{1}{(2n+1)!} 
	\begin{pmatrix}O_{1} & (-1)^{n} \sqrt{BB^{*}}^{2n}B \cr
	(-B^{*})(-1)^{n}\sqrt{BB^{*}}^{2n} &  O_{2}\end{pmatrix}$$ 
	$$    = \begin{pmatrix}O_{1} &  \sqrt{B^{*}B}^{-1} \sin \sqrt{B^{*}B} B \cr
	(-B^{*})\sqrt{B^{*}B}^{-1} \sin \sqrt{B^{*}B} &  O_{2}\end{pmatrix}.$$ 
\end{rem}
Let $X = B \frac{\sin \sqrt{B^{*}B}}{\sqrt{B^{*}B}}$ and $Y = \cos \sqrt{B^{*}B}$.  Then
$$    X^{*}X = \frac{\sin \sqrt{B^{*}B}}{\sqrt{B^{*}B}} B^{*}B \frac{\sin \sqrt{B^{*}B}}{\sqrt{B^{*}B}} = \sin
^{2} \sqrt{B^{*}B},$$ 
and\newline
$$    X^{*}X + Y^{2} = \sin ^{2} \sqrt{B^{*}B} + \cos ^{2} \sqrt{B^{*}B} = I, \  Y = (I - X^{*}X)^{1/2}.$$ 
In the same way, we have\newline
$$    XX^{*} = B \frac{\sin \sqrt{B^{*}B}}{\sqrt{B^{*}B}} \frac{\sin \sqrt{B^{*}B}}{\sqrt{B^{*}B}} B^{*}$$ 
$$     = B \frac{\sin ^{2} \sqrt{B^{*}B}}{B^{*}B} B^{*} = \frac{\sin ^{2} \sqrt{BB^{*}}}{BB^{*}} BB^{*} =
\sin ^{2} \sqrt{BB^{*}},$$   
where we used the fact that for an entire function $f(x) = \sum _{n=0}^{\infty } a_{n}
x^{n}$,\newline
$$    B f(B^{*}B)B^{*} = \sum _{n=0}^{\infty } a_{n}B(B^{*}B)^{n}B^{*} = \sum _{n=0}^{\infty } a_{n}(BB^{*})^{n} BB^{*} =
f(BB^{*})BB^{*},$$ 
and $\frac{\sin \sqrt{B^{*}B}}{\sqrt{B^{*}B}}$ is an entire function of $B^{*}B$.  This shows
\newline
$$    \cos \sqrt{BB^{*}} = (I - XX^{*})^{1/2}.$$ 
Since $K(B) \in  \mathfrak{u}(n)$, $\exp K(B) \in  U(n)$ and 
\beq  \label{li:exp}
\exp K(B) =
\begin{pmatrix} (I_{k_1} - XX^{*})^{1/2} & X \cr
	- X^{*} & (I_{k_2} - X^{*}X)^{1/2} \end{pmatrix} = W(X) \in  U(n).
\eeq
Without knowing such background, we can show directly the
unitarity of the matrix (\ref{li:exp}).
\begin{prop}
	For $X \in  M(k_{1}, k_{2})$, $X^{*}X \leq  I_{k_2} \Leftrightarrow  XX^{*} \leq  I_{k_1}$.
\end{prop}
\begin{proof} Here is the elementary proof.
	$$    X^{*}X \leq  I_{k_2} \Leftrightarrow  \forall e \in  \mathbb C^{k_{2}}(\Vert e\Vert  = 1 \Rightarrow  (e, X^{*}Xe)_{2} \leq  (e, I_{k_2}e)_{2} = 1)$$ 
	$$    \Leftrightarrow  \forall e \in  \mathbb C^{k_{2}}(\Vert e\Vert  = 1 \Rightarrow  (Xe, Xe)_{1} \leq  1)
	\Leftrightarrow  \forall e \in  \mathbb C^{k_{2}}(\Vert e\Vert  = 1 \Rightarrow  \Vert Xe\Vert _{1} \leq  1)$$ 
	$$    \Leftrightarrow  \forall e \in  \mathbb C^{k_{2}}, \forall d \in  \mathbb C^{k_{1}}(\Vert e\Vert  = \Vert d\Vert  = 1 \Rightarrow  (d, Xe)_{1} \leq  1) $$ 
	$$    \Leftrightarrow  \forall e \in  \mathbb C^{k_{2}}, \forall d \in  \mathbb C^{k_{1}}(\Vert e\Vert  = \Vert d\Vert  = 1 \Rightarrow  (X^{*}d, e)_{2} \leq  1) $$ 
	$$    \Leftrightarrow  \forall d \in  \mathbb C^{k_{1}}(\Vert d\Vert  = 1 \Rightarrow  \Vert X^{*}d\Vert _{2} \leq  1)
	\Leftrightarrow  \forall d \in  \mathbb C^{k_{1}}(\Vert d\Vert  = 1 \Rightarrow  (X^{*}d, X^{*}d)_{2} \leq  1) $$ 
	$$    \Leftrightarrow  \forall d \in  \mathbb C^{k_{1}}(\Vert d\Vert  = 1 \Rightarrow  (d, XX^{*}d)_{1} \leq  1) \Leftrightarrow  XX^{*} \leq  I_{k_1}.$$ 
\end{proof}
\begin{prop}  \label{li:WX}
	Let $X \in  \bar{B} (k_{1}, k_{2})$ = $\{ X \in  M(k_{1}, k_{2}); X^{*}X \leq  I_{k_2}\} $.  Then $(I_{k_2} -
	X^{*}X)^{1/2}$ and $(I_{k_1} - XX^{*})^{1/2}$ are well defined, and $W(X)$ of (\ref{li:exp}) which appeared in Proposition \ref{gr:lambda} is unitary.
\end{prop}
\begin{proof}
	$$ 
	\begin{pmatrix} (I_{k_1} - XX^{*})^{1/2} & -X \cr
	X^{*} & (I_{k_2} - X^{*}X)^{1/2} \end{pmatrix}
	\begin{pmatrix} (I_{k_1} - XX^{*})^{1/2} & X \cr
	- X^{*} & (I_{k_2} - X^{*}X)^{1/2} \end{pmatrix}$$ 
	$$    = \begin{pmatrix} I_{k_1} - XX^{*} + XX^{*} & (I_{k_1} - XX^{*})^{1/2}X - X(I_{k_2} - X^{*}X)^{1/2} \cr
	X^{*}(I_{k_1} - XX^{*})^{1/2} - (I_{k_2} - X^{*}X)^{1/2}X^{*} & X^{*}X + I_{k_2} - X^{*}X \end{pmatrix}.$$ 
	Since $(1 + x)^{\alpha }$ for $\alpha  = 1/2$ is expanded as $1 + \sum _{n=1}^{\infty } {\alpha  \choose n} x^{n}$ for $-1 < x < 1$,\newline
	$$    (I_{k_1} - XX^{*})^{1/2}X = (I_{k_1} + \sum _{n=1}^{\infty } {\alpha  \choose n} (-XX^{*})^{n})X $$ 
	$$    = X + \sum _{n=1}^{\infty } {\alpha  \choose n} X(-X^{*}X)^{n} = X(I_{k_1} - X^{*}X)^{1/2}$$ 
	and\newline
	$$    X^{*}(I_{k_1} - XX^{*})^{1/2} = X^{*}(I_{k_2} + \sum _{n=1}^{\infty } {\alpha  \choose n} (-XX^{*})^{n}) $$ 
	$$    = X^{*} + \sum _{n=1}^{\infty } {\alpha  \choose n} (-X^{*}X)^{n}X^{*} = (I_{k_2} - X^{*}X)^{1/2}X^{*}$$ 
	hold for $X \in  B(k_{1}, k_{2})$ = $\{ X \in  M(k_{1}, k_{2}); X^{*}X < I_{k_2} \} $. 
	In order to show the above two formulae for $X \in  \bar{B} (k_{1}, k_{2})$,
	we employ a limiting process $B(k_{1},
	k_{2}) \ni  X_{n} \rightarrow  X \in  \bar{B} (k_{1}, k_{2})$ as $n \rightarrow  \infty $ with respect to a norm $\Vert X\Vert ^{2} =
	{\rm Tr\, }X^{*}X = \sum _{i=1}^{k_{1}}\sum _{j=1}^{k_{1}} \vert x_{ij}\vert ^{2}$.  Thus we have $W(X)^* W(X) = I_n$ for all $X \in  \bar{B} (k_{1}, k_{2})$.
\end{proof}
{
	\begin{rem}  \label{li:XX*}
		The following relation is useful. Here $\alpha=1/2$.
		$$    (I_{k_1} - XX^{*})^{1/2} = I_{k_1} - X\sum _{n=1}^{\infty } {\alpha \choose n} (-X^{*}X)^{n-1} X^{*} $$
		$$  = I_{k_1} + X(X^{*}X)^{-1/2} \sum _{n=1}^{\infty } {\alpha  \choose n} (-X^{*}X)^{n} (X^{*}X)^{-1/2} X^{*}$$ 
		$$  = I_{k_1} + X(X^{*}X)^{-1/2} [(1 - X^{*}X)^{1/2} - 1] (X^{*}X)^{-1/2} X^{*}.$$ 
	\end{rem}
}
\begin{rem}  Since a Lie algebra describes only the local properties of its Lie group, the mapping $B \rightarrow \pi (\exp K(B))$ gives a local homeomorphism, that is,
	there is a neighborhood $V$ of $0$ in $M(k_{1}, k_{2})$ and $U$ of $e$ in $U(n)$ such that $V \ni B \rightarrow \pi (\exp K(B)) \in \pi (U)$ is homeomorphic (see \cite{Ma72}).  But for $W(X)$,
	Proposition \ref{gr:kappa} says that the mapping
	$\kappa _{\sigma }: B(k_{1}, k_{2}) \ni X \rightarrow \pi (W(X) \in \Omega _{\sigma }\subset \pi (U(n))$ is bijective and
	Proposition \ref{bd:section} says the mapping
	$\bar{\kappa }: \bar{B}(k_{1}, k_{2}) \ni B \rightarrow \pi (W(X) \in \pi (U(n))$ is surjective.
\end{rem}

\section{Examples} \label{example}
Here we give two examples of the parametrization of degenerate
density matrices with diagonal matrices of eigenvalues of the forms:\newline
1)  $D_{4}(\boldsymbol{\lambda }) = {\rm diag \, }_{4}(\lambda _{1}I_{3}, \lambda _{2}I_{1})$,\newline
2)  $D_{4}(\boldsymbol{\lambda }) = {\rm diag \, }_{4}(\lambda _{1}I_{2}, \lambda _{2}I_{2})$.\newline
For the first case, the density matrices are parametrized by
$$      \Lambda_2^{\neq}  \times  G(1, \mathbb C^{4}) = \Lambda_2^{\neq} \times  CP^{3}$$
(see (\ref{gr:degenerateden})).\newline
Since $\Omega _{4} \subset  CP^{3}$ is an open dense subset of $CP^{3}$ and $\Omega _{4}$ is
parametrized by $B(3) = \{ z \in  \mathbb C^{3}; \vert z\vert  < 1\} $, almost all density
matrices are parametrized by $\Lambda  \times  B(3)$.\newline
Concretely, we have the following parametrization:
$$      \Lambda_2^{\neq}  \times  B(3) \ni  ((\lambda _{1}, \lambda _{2}), x) \rightarrow  \rho  = $$
$$  \begin{pmatrix} (I_{3} - xx^{*})^{1/2} & x \cr -x^{*} & (1 - x^{*}x)^{1/2}
\end{pmatrix}
\begin{pmatrix} \lambda _{1}I_{3} & 0 \cr 0 & \lambda _{2}I_{1} \end{pmatrix}
\begin{pmatrix} (I_{3} - xx^{*})^{1/2} & x \cr -x^{*} & (1 - x^{*}x)^{1/2}
\end{pmatrix} ^{*}. $$ 
$(I_{3} - xx^{*})^{1/2}$ can be calculated as:
$$     (I_{3} - xx^{*})^{1/2} = I_{3} + (x^{*}x)^{-1} [(1 - x^{*}x)^{1/2} - 1]xx^{*}$$
(see Remark \ref{li:XX*}).

For the second case, the density matrices are parametrized by
$$      \Lambda_2^{\neq} \times  G(2, \mathbb C^{4}).$$
Since $\Omega _{e} \subset  G(2, \mathbb C^{4})$ is an open dense subset of $G(2, \mathbb C^{4})$ and $\Omega _{e}$ is
parametrized by $B(2,2) = \{ X \in  M(2,2); X^{*}X < I_{2}\} $, almost all density
matrices are parametrized by $\Lambda  \times  B(2,2)$.\newline
Concretely, we have the following parametrization:
$$      \Lambda_2^{\neq} \times  B(2,2) \ni  ((\lambda _{1}, \lambda _{2}), x) \rightarrow  \rho  = $$
$$  \begin{pmatrix} (I_{2} - xx^{*})^{1/2} & x \cr -x^{*} & (I_{2} - x^{*}x)^{1/2}
\end{pmatrix}
\begin{pmatrix} \lambda _{1}I_{2} & 0 \cr 0 & \lambda _{2}I_{2} \end{pmatrix}
\begin{pmatrix} (I_{2} - xx^{*})^{1/2} & x \cr -x^{*} & (I_{2} - x^{*}x)^{1/2}
\end{pmatrix} ^{*}. $$ 
But unfortunately, $(I_{2} - xx^{*})^{1/2}$ and $(I_{2} - x^{*}x)^{1/2}$ are not easy
to calculate.  So, we employ Corollary \ref{pr2:canonicaldec} which states that there is a
bijection $\phi _{e}: \Omega _{e} \rightarrow  \Omega _{2,4} \times  \Omega _{2,3} \subset  CP^{3} \times  CP^{2}$, and $\psi \circ \phi _{e}$ is a local
section of $U(4)$ on $\Omega _{e}$, where
$$    \psi (z_{2}, z_{1}) = \iota _{4}(z_{2}) \begin{pmatrix} \iota _{3}(z_{1}) & 0 \cr 0 & I_{1}
\end{pmatrix}, $$ 
and $\iota _{j}$ is the section defined by (\ref{pr:section}).
Concretely, we have the following parametrization:
$$      \Lambda_2^{\neq} \times  B(2)^{2} \ni  (\lambda _{1}, \lambda _{2}, x_{2}, x_{1}) \rightarrow  \rho  = U(x_{2}, x_{1}) {\rm diag }_{4}(\lambda _{1}I_{2}, \lambda _{2}I_{2}) U(x_{2}, x_{1})^{*},$$
$$  U(x_{2}, x_{1}) = \begin{pmatrix} (I_{2} - x_{2}x_{2}^{*})^{1/2} & 0 & x_{2} \cr
0 & 1 & 0  \cr -x_{2}^{*} & 0 & (1 - x_{2}^{*}x_{2})^{1/2} \end{pmatrix}
\begin{pmatrix} (I_{2} - x_{1}x_{1}^{*})^{1/2} & x_{1} & 0 \cr 
-x_{1}^{*} & (1 - x_{1}^{*}x_{1})^{1/2} & 0 \cr
0 & 0 & 1 \end{pmatrix} . $$  

\section{Conclusion}
The problem of parametrizing degenerate density matrices required to develope{a new} { approach using techniques from the theory of homogeneous spaces} as outlined in sections \ref{intro} - \ref{section}. This approach is not based on the use of Lie algebra methods. Actually our approach helps to detect some short comings of the Lie algebra approach as used the the given references{, i.e., the exponential map from Lie algebra to Lie group is not one to one and onto,} and to correct these, also in the case of non-degenerate density matrices. These short comings are due to the non-injectivity of the given map at the boundary of the respective parameter domain. 

\section*{Acknowledgements} We are very grateful to the National Institute for Theoretical Physics (South Africa) which helped to realise this collaboration through a substantial grant from their visitor program 2016 for S. Nagamachi.

\end{document}